\documentclass[12pt]{article}
\usepackage{a4,amsmath,amssymb,amsthm,amscd,cite,graphicx}
\usepackage{verbatim,numprint,siunitx,mathrsfs,esint,xcolor}
\usepackage{rotating}
\usepackage{hyperref} \hypersetup{colorlinks, linkcolor=red }
\usepackage{tikz}
\usetikzlibrary{intersections,calc,matrix,arrows, decorations.markings,shapes}
\usepackage{multirow}
\newtheorem{theorem}{Theorem}[section]
\newtheorem{prop}[theorem]{Proposition}

\newtheorem*{T1}{Proposition~\ref{modularpropM24}}

\def\Bbb{\mathbb} \def\BZ{\Bbb Z}

 \newcommand{\be}{\begin{equation}}
  \newcommand{\ee}{\end{equation}}

\catcode`@=11 \@addtoreset{equation}{section} \catcode`@=12

\begin{document}

\begin{titlepage}
\begin{flushright}
Jan 2021 (v2.1)\\
\texttt{arXiv:2011.07922 [hep-th]}
\end{flushright}
\begin{center}
\textsf{\large Mathieu Moonshine and Siegel Modular Forms}\\[12pt]
Suresh Govindarajan$^{\dagger}$ and Sutapa Samanta$^{*}$ \\
$^\dagger$Department of Physics,
Indian Institute of Technology Madras ,
Chennai 600036 INDIA\\
 $^*$School of Physical Sciences, Indian Association for the Cultivation of Science, Jadavpur, Kolkata 700032, INDIA\\[4pt]
Email: $^\dagger$suresh@physics.iitm.ac.in, $^*$psss2238@iacs.res.in
\end{center}
\begin{abstract}
A second-quantized version of Mathieu moonshine leads to product formulae for functions that are potentially genus-two Siegel Modular Forms analogous to the Igusa Cusp Form. The modularity of these functions do not follow in an obvious manner. For some conjugacy classes, but not all, they match known modular forms.  In this paper,
we express the product formulae for all conjugacy classes of $M_{24}$ in terms of products of standard modular forms. This provides a new proof of their modularity.
\end{abstract}
\end{titlepage}

\section{Introduction}

Modular forms are important ingredients in string theory as well as in mathematical physics. The generating function of a $\frac12$-BPS state of $\mathcal{N}=4$ superstring theory obtained by compactifying the heterotic string on a six-torus is a modular form. Using dualities, the $\mathcal{N}=4$ superstring theory in four dimensions can also be obtained by compactifying type II  string theory on $K3\times T^2$. The generating function of the degeneracy of electrically charged $\tfrac12$-BPS state is given by the twenty fourth power of the Dedekind $\eta$-function\cite{Dijkgraaf:1996it}. When the six-torus is twined by a symmetry element, then the generating function is given by a multiplicative $\eta$ product for a cycle shape that is associated with a conjugacy class of $M_{24}$\cite{Govindarajan:2008vi}. Given a cycle shape $\rho$, the  $\eta$-product associated with the cycle shape are defined by the following map:
$$ \rho = 1^{a_1} 2^{a_2}\cdots n^{a_n} \quad \longrightarrow\quad \eta_\rho (\tau) = \eta(\tau)^{a_1}\eta(2\tau)^{a_2}\cdots \eta(n\tau)^{a_n}\ .$$ A cycle shape is called balanced if there exists a positive integer $M$ such that
$$\rho =1^{a_1} 2^{a_2}\cdots n^{a_n}=\left(\tfrac{M}1\right)^{a_1} \left(\tfrac{M}2\right)^{a_2}\cdots \left(\tfrac{M}n\right)^{a_n}.$$
Conway and Norton \cite{Conway1979} observed that all the conjugacy classes of the Mathieu group $M_{24}$ are given by a balanced cycle shape.  Thus there is a connection between the generating function of the $\frac12$ BPS states and the conjugacy classes of the Mathieu group. This is known as Mathieu moonshine \cite{Dummit:1985, Mason:1985, Govindarajan:2008vi}.

Another Mathieu moonshine was discovered by Eguchi, Ooguri, and Tachikawa \cite{Eguchi:2010ej} that relates the conjugacy class $1a$ of $M_{24}$ with the elliptic genus of $K3$. This elliptic genus is a weight zero and index one Jacobi form. Later, the Jacobi forms for all other conjugacy classes are constructed \cite{Cheng:2010pq, Gaberdiel:2010ch, Eguchi:2010fg}. Gannon proved the existence of a class function that evaluates to all these Jacobi forms\cite{Gannon:2012ck}.  We list the cycle shapes and the conjugacy classes of $M_{24}$ in Appendix \ref{EOT}.

Now, if we consider the $\frac14$-BPS states of the same $\mathcal{N}=4$ superstring theory as mentioned earlier, the generating functions of the degeneracy are genus-two Siegel Modular Forms in some cases\cite{Govindarajan:2008vi}. In \cite{Govindarajan:2008vi}, the authors  proposed the following map relating cycle shapes  to Siegel Modular Forms:
\[
\rho=1^{a_1} 2^{a_2}\cdots n^{a_n} \quad \longrightarrow\quad \Phi_k^{\rho}(\mathbf{Z})\ ,
\]
where $\Phi_k^{\rho}(\mathbf{Z})$ is a Siegel Modular Form with weight $k=(-2+\sum_i a_i)$ of a suitable subgroup of $Sp(4,\mathbb{Z})$.
For conjugacy class $1a$, this is the well-known weight ten Igusa cusp form\cite{Dijkgraaf:1996it}.  The two Mathieu moonshines were combined in a second-quantized version leading to  class function that implied a product formula for $\Phi_k^{\rho}(\mathbf{Z})$ for all conjugacy classes for $M_{24}$ \cite{Govindarajan:2011em}. The proof that these product formulae are indeed Siegel Modular Forms for all conjugacy classes first appeared in \cite{Persson:2013xpa}. It was also proven for a sub-class that appear in the context of $L_2(11)$ moonshine in \cite{Govindarajan:2018ted}. The square-root of some of these Siegel Modular Forms also appear in the denominator formula of Borcherds-Kac-Moody Lie super-algebra \cite{Govindarajan:2008vi, Govindarajan:2009qt, Govindarajan:2010fu, Govindarajan:2018ted, Govindarajan:2019ezd}.

In this paper, we provide a new and distinct proof that the $\Phi_k^{\rho}(\mathbf{Z})$ are Siegel Modular Forms. The various connections are described in Fig. \ref{M24Moonshine}. 
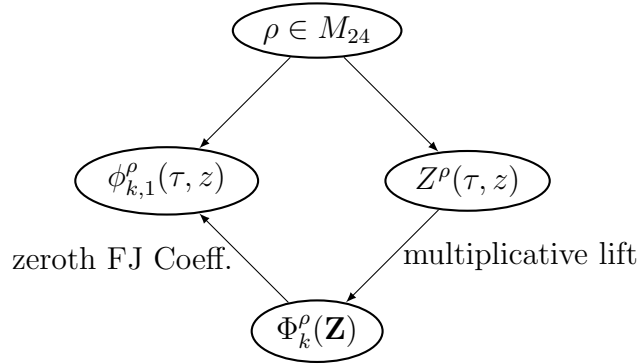
\begin{figure}[ht]
\centering
\begin{tikzpicture}
\draw (0,0)node[ellipse, thick, inner sep=2pt,draw](y1){$\rho\in M_{24}$};
\draw (2,-2)node[ellipse, thick, inner sep=2pt,draw](y2){$Z^{\rho}(\tau,z)$};
\draw[-latex](y1)--(y2);
\draw (-2,-2)node[ellipse, thick, inner sep=2pt,draw](y3){$\phi_{k,1}^{\rho}(\tau,z)$};
\draw[-latex](y1)--(y3);
\draw (0,-4)node[ellipse, thick, inner sep=2pt,draw](y4){$\Phi_{k}^{\rho}(\mathbf{Z})$};
\draw[-latex](y2) to node[right]{multiplicative lift}(y4);
\draw[-latex](y4) to node [left]{zeroth FJ Coeff.}(y3);
\end{tikzpicture}
\caption{Moonshine for $M_{24}$. $\phi_{k,1}^{\rho}(\tau,z)$ is constructed from the multiplicative $\eta$-products and $Z^{\rho}(\tau,z)$ is the twined elliptic genus. $\Phi_{k}^{\rho}(\mathbf{Z})$ is the genus-two Siegel Modular Form.}
\label{M24Moonshine}
\end{figure}

 \noindent The main result of this paper is the following proposition  that we prove. 
\begin{T1}
 For $g\in M_{24}$, let $\rho_m=[g^m]$, $\rho=\rho_1$, $Z_{0,1}^{\rho_m}(\tau,z)$ be the Jacobi form of weight zero and index one corresponding to $g^m$ and $\Phi_k^{\rho}(\mathbf{Z})$ be the modular form defined by the multiplicative lift given in Eq. \eqref{productformulamain}. Then
\begin{equation*}
\Big(\Phi_k^{\rho}(\mathbf{Z})\Big)^{p_{\rho}} = \prod_{m|p_\rho} \left(B_{\frac{p_\rho}m Z^{\rho_m}}(m\,\mathbf{Z})\right) \ ,
\end{equation*}
where $p_{\rho}$ is the length of the shortest cycle in the cycle shape for the conjugacy class $\rho$ and $B_\psi(\mathbf{Z})$ is the Siegel Modular Form given by Theorem \ref{CGproduct}.
\end{T1}
\noindent \textbf{Remark:} Our proof differs from the one by Persson-Volpato\cite{Persson:2013xpa} in that we show that the product formula for the conjugacy class of an order $N$ element of $M_{24}$ correspond to Siegel Modular Forms at level $N$ of $Sp(4,\BZ)$ with character corresponding to a $p_\rho$-th root of unity, where $p_\rho$ is as defined above. In \cite{Persson:2013xpa}, it was shown that these products correspond to subgroup at higher level $p_\rho N$.\footnote{We thank Daniel Persson and Roberto Volpato for a useful email correspondence.} Thus our result complements theirs. Our proof also explicitly connects to Borcherds products that appear in number theory. In the process we obtain a simple looking formula that holds for all conjugacy classes.

After the introductory section, in Section 2, we discuss the ideas leading to  Proposition \ref{modularpropM24} and a sketch of its proof using two illustrative examples. Section 3 summarises our results and concluding remarks. Appendix A provides the definitions and examples of modular forms and Jacobi forms. Appendix B lists the Jacobi forms for all conjugacy classes of $M_{24}$. A fairly long appendix D provides the details of the proof of Proposition \ref{modularpropM24}.

\section{Construction of Siegel Modular Forms for $M_{24}$}

In this section, we construct genus two Siegel Modular Forms for all the conjugacy classes of $M_{24}$. We know that there's a weight zero and index one Jacobi form (see Appendix \ref{EOT}) \cite{Eguchi:2010ej, Cheng:2010pq, Gaberdiel:2010ch, Eguchi:2010fg} for every conjugacy class of $M_{24}$. We apply the multiplicative lift to them to construct the Siegel Modular Forms. This multiplicative lift naturally leads to a product formula. We prove the modularity of the construction in two ways -- (i) one is to construct the sum side using the additive lift and showing that this is equal to the product formula, (ii) the other is to compare the product formula with the Borcherds products. The first method works in most cases, while the second method works for all the cases.

\subsection{The multiplicative lift}

The Mathieu moonshine leads to a Siegel Modular Form which is denoted by $\Phi_k^{\rho}(\mathbf{Z})$. The formula for $\Phi_k^{\rho}(\mathbf{Z})$ is given by \cite{Govindarajan:2011em}
\begin{align}\label{multiplicativelift}
\Phi_k^{\rho} = s\, \phi^{\rho}_{k,1}\times \exp\left(-\sum_{m=1}^\infty s^m V_m \cdot Z_{0,1}^{\rho}(\tau, z)\right)\ ,
\end{align}
where $V_m$ is the generalized Hecke-like operator. This operator acts on a weight zero and index one Jacobi form as follows
\begin{align}\label{gen_hecke}
V_m \cdot Z_{0,1}^{\rho} (\tau,z) \equiv \frac1m\sum_{ad=m}\sum_{b=0}^{d-1} Z_{0,1}^{\rho_a}\left( \frac{a\tau+b}{d}, az\right)\ .
\end{align}
where  $\rho$ denotes the conjugacy class of the element $g$ of order $N$. Then $\rho_a$ denotes the conjugacy class of the element $g^a$. We set $\rho_0=1^{24}$. The Fourier-Jacobi coefficients $f_a(n,\ell)$ of $Z^{\rho_a}_{0,1}(\tau,z)$ is defined as follows:
\begin{align}
Z_{0,1}^{\rho_a}(\tau,z) =\sum_{n=0}^\infty \sum_{\ell\in\mathbb{Z}} f_a(n,\ell) q^n r^\ell\ .
\end{align}
Let us now define $c^\alpha(n,\ell)$ via following discrete Fourier transform.
\begin{align}\label{FJcoeff}
f_a(n,\ell) = \sum_{\alpha=0}^{N-1}\left( \omega_{N}^{a} \right)^\alpha c^\alpha(n,\ell)\ ,
\end{align}
where $\omega_{N}=\exp(2\pi i /N)$.
The inverse Fourier transform of Eq. \eqref{FJcoeff} is given by
\begin{align}\label{inverseFJ}
c^\alpha = \frac1N \sum_{a=0}^{N-1}  \exp(-2\pi i a \alpha/N) f_a\ .
\end{align}
Using \eqref{gen_hecke}, we obtain a product formula for \eqref{multiplicativelift} as follows:
\begin{align}\label{productformulamain} 
\Phi_k^{\rho}(\mathbf{Z})= s\, \phi_{k,1}^{\rho} \times \prod_{m=1}^{\infty} \prod_{\alpha=0}^{N-1}\prod_{n=0}^\infty \prod_{\substack{\ell\in \mathbb{Z}\\ 4nm-\ell^2\ge0}} \left(1-\omega_n^\alpha q^{n}r^{\ell} s^{m}\right)^{c^\alpha(nm,\ell)}\ ,
\end{align}
where the zeroth Fourier-Jacobi term, i.e., the coefficient of $s^1$ in the Fourier expansion is given by
$$\phi_{k,1}^{\rho}(\tau,z) =\frac{\vartheta_1(\tau,z)^2}{\eta(\tau)^6}\ \eta_\rho(\tau)\ .$$
\textbf{Remark:} Let $p_\rho$ (resp. $N_\rho$) be the length of the smallest (resp. largest) cycle in $\rho$. From the result of Cheng and Duncan\cite{Cheng:2011}, one has that $\eta_\rho(\tau)^{p_\rho}$ (and hence $[\phi_{k,1}^{\rho}(\tau,z)]^{p_\rho}$) is a modular (Jacobi) form for the subgroup $\Gamma_0(N_\rho)$ of $SL(2,\BZ)$. We thus anticipate that $[\Phi_k^{\rho}(\mathbf{Z})]^{p_\rho}$ will be a Siegel Modular Form for a level $N_\rho$ subgroup of $Sp(4,\BZ)$.

For the cases when the order of $g (\in M_{24})$ is prime (i.e., $N=2,3,5,7,11,23$), one has the conjugacy class $g^a$ for $a\ne 0$ mod $N$ is the same as that of $g$. Thus, one has $\rho_a=\rho$ for $a\ne0$ mod $N$. Now using Eq.~\eqref{inverseFJ} we get
$$ c^0 = \frac1N(f_0 + (N-1)f_1)\hspace{0.3cm}\text{and}\hspace{0.3cm} c^1= \frac1N(f_0-f_1)\ .$$
For these cases, on using the product representation
for the theta and $\eta$-functions that appear in $\phi^\rho_{k,1}(\tau,z)$ and using $\prod_{n=0}^{N-1} (1- \omega^n x) = (1-x^N)$, the formula given in Eq. \eqref{productformulamain} simplifies to
\begin{align}
\Phi_k^{\rho} &= qrs\prod_{(n,\ell,m)>0} \Big(1-q^n r^\ell s^m\Big)^{c^0(nm,\ell) - c^1(nm,\ell)} \Big(1-q^{nN} r^{\ell N} s^{m N}\Big)^{c^1(nm,\ell)}\ ,\\
&=qrs\prod_{(n,\ell,m)>0} \Big(1-q^n r^\ell s^m\Big)^{f_1(nm,\ell)} \Big(1-q^{nN} r^{\ell N} s^{mN}\Big)^{\frac{f_0(nm,\ell)- f_1(nm,\ell)}N}\ ,
\end{align}
where $(n,\ell,m)>0$ implies $n>0,$ or $n=0$ and $m>0$, or $n=m=0$ and $\ell<0$.
We will treat the cases of composite $N$ separately.

\subsection{The additive lift}

The another construction of $\Phi^{\rho}(\mathbf{Z})$ is via additive lift. The additive seed is the weak Jacobi form of weight $k$ and index $t$ for the sub-group $\Gamma_0(Nq,q)$. This index $t$ can be integral or half integral. Let us consider a Jacobi form of weight $k$ and index $t$, $\phi_{k,t} (\tau,z) $. Then for $k>0$, the additive lift is given by
\begin{align}\label{additivelift}
\Phi^{\rho}(\mathbf{Z}) = \sum_{m=1\mod q} s^{mt}\phi^\rho|_{k,t} T_-(m)(\tau,z)\ ,
\end{align}
where $T_-(m)$ is the Hecke-like operator introduced by Cl\'ery and Gritsenko \cite{Gritsenko:2008} and is given by
\begin{align}\label{CGHecke}
T_-(m) := \sum_{\substack{ad=m\\ (a,N)=1\\ a>0}} \Gamma_1(Nq,q)\, \sigma_a \begin{pmatrix} a & bq\\ 0&d\end{pmatrix}\ ,
\end{align}
where $\sigma_a = \left(\begin{smallmatrix} a^{-1} & 0\\ 0 & a\end{smallmatrix}\right) \mod Nq$ and $(m,q)=1$. For all $M_{24}$ conjugacy classes ($\rho$), we have $t=1$ and $q=1$.

\subsection{Modularity by comparing with the additive side}

The additive seed for these cases are Jacobi forms of weight $k$ and index one,
\begin{align}
\phi_{k,1}^\rho (\tau,z) = \frac{\vartheta_1(\tau,z)^2}{\eta(\tau)^6}\ \eta_\rho (\tau) = \sum_{n,\ell} a(n,\ell) q^{n} r^\ell\ .
\end{align}
A necessary  but not sufficient condition for the compatibility of the additive lift with the multiplicative lift is:
\begin{align}\label{compat}
\left.\left[ \frac{\theta_1(\tau,z)^2}{\eta(\tau)^6} \eta_\rho (\tau)\right]\right|_{k,1} T_-(2)(\tau,z) =- Z^\rho (\tau,z) \left[ \frac{\theta_1(\tau,z)^2}{\eta(\tau)^6} \eta_\rho (\tau)\right]\ .
\end{align}
This is the coefficient of $s$ in both the lifts, i.e., the ones given in Eq. \eqref{multiplicativelift} and Eq. \eqref{additivelift}. This condition holds for all the cycle shapes of $M_{24}$ except for the following eight cycle shapes where we observe that \cite{Eguchi:2011aj}
\begin{align}
&-\frac{T_2 \phi^{1^2 11^2}(\tau,z)}{\phi^{1^2 11^2}(\tau,z)} -Z^{1^2 11^2} (\tau,z)= \frac{11}{2}\phi^{1^2 11^2}(\tau,z) \ ,\nonumber\\
&-\frac{T_2 \phi^{1^12^1 7^1 14^1}(\tau,z)}{\phi^{1^12^1 7^1 14^1}(\tau,z)} - Z^{1^12^1 7^1 14^1} (\tau,z)=7\phi^{1^12^1 7^1 14^1}(\tau,z) \ ,\nonumber\\
&-\frac{T_2 \phi^{1^13^1 5^1 15^1}(\tau,z)}{\phi^{1^13^1 5^1 15^1}(\tau,z)} - Z^{1^13^1 5^1 15^1} (\tau,z)=\frac{15}2\phi^{1^13^1 5^1 15^1}(\tau,z) \ ,\\
&-\frac{T_2 \phi^{2^2 10^2}(\tau,z)}{\phi^{2^2 10^2}(\tau,z)} - Z^{2^2 10^2} (\tau,z)=10\phi^{2^2 10^2}(\tau,z) \ ,\nonumber\\
&-\frac{T_2 \phi^{2^1 4^1 6^1 12^1}(\tau,z)}{\phi^{2^1 4^1 6^1 12^1}(\tau,z)} - Z^{2^1 4^1 6^1 12^1} (\tau,z)=12\phi^{2^1 4^1 6^1 12^1}(\tau,z) \ , \nonumber
\end{align}
and 
\begin{align}
&-\frac{T_2 \phi^{4^6}(\tau,z)}{\phi^{4^6}(\tau,z)} - Z^{4^6} (\tau,z)=16\frac{\eta(2\tau)^4 \eta(8\tau)^4}{\eta(4\tau)^4}\frac{\theta_1(\tau,z)^2}{\eta(\tau)^6} \ ,\nonumber\\
&-\frac{T_2 \phi^{3^8}(\tau,z)}{\phi^{3^8}(\tau,z)} - Z^{3^8} (\tau,z)=18\frac{\eta(\tau)^3 \eta(9\tau)^3}{\eta(3\tau)^2}\frac{\theta_1(\tau,z)^2}{\eta(\tau)^6} \ ,\\
&-\frac{T_2 \phi^{6^4}(\tau,z)}{\phi^{6^4}(\tau,z)} - Z^{6^4} (\tau,z)= 2\left(\eta_{6^4}(\tau)+ \frac{\eta(\tau)^3 \eta(9\tau)^3}{\eta(3\tau)^2}+ 6 \frac{\eta(2\tau)^3 \eta(18\tau)^3}{\eta(6\tau)^2}\right.\nonumber\\
&\hspace{6cm}\left. + \frac{\eta(4\tau)^3 \eta(36\tau)^3}{\eta(12\tau)^2}+\cdots\right)\frac{\theta_1(\tau,z)^2}{\eta(\tau)^6} \ ,\nonumber
\end{align}
where $T_2\phi^\rho$ is short for $\phi^\rho|_{k,1} T_-(2)$. All the above Jacobi forms except the ones correspond to $3^8$ and $4^6$ are of weight $k=0$ and we have na\"ively applied a formula which assumes $k>0$. 

Even for the cases for which the compatibility condition Eq. \eqref{compat} holds, we need to prove that all other terms also match. We do not pursue this method here. We prove the modularity using another method which we discuss next.

\subsection{Modularity by comparing with a Borcherds formula}

We begin with a theorem due to Cl\'ery-Gritsenko (see also \cite{GritsenkoNikulinI,GritsenkoNikulinII,Aoki:2005}) that leads to a Borcherds product formula for a meromorphic Siegel Modular Form starting from a nearly holomorphic Jacobi form of weight zero and index $t$.

\begin{theorem}[Cl\'ery-Gritsenko\cite{Gritsenko:2008}] \label{CGproduct} Let $\psi$ be a nearly holomorphic Jacobi form of weight $0$ and index $t$ of $\Gamma_0(N)$. Assume that for all cusps of $\Gamma_0(N)$ one has $\frac{h_e}{N_e} c_{f/e}(n,\ell)\in \mathbb{Z}$ if $4nt -\ell^2\leq 0$. Then the product
\[
B_\psi(\mathbf{Z}) = q^A r^B s^C \prod_{f/e\in \mathcal{P}} \prod_{\substack{n,\ell,m\in \mathbb{Z}\\ (n,\ell,m)>0}} \Big(1-(q^nr^\ell s^{tm})^{N_e}\Big)^{\frac{h_e}{N_e} c_{f/e}(nm,\ell)}\ ,
\]
with
\[
A=\frac1{24} \sum_{\substack{f/e\in \mathcal{P}\\ \ell \in \mathbb{Z}}} h_e\, c_{f/e}(0,\ell) ,\
B=\frac1{2} \sum_{\substack{f/e\in \mathcal{P}\\ \ell \in \mathbb{Z}_{>0}}} \ell h_e\, c_{f/e}(0,\ell),\
C=\frac1{4} \sum_{\substack{f/e\in \mathcal{P}\\ \ell \in \mathbb{Z}}} \ell^2 h_e\, c_{f/e}(0,\ell)\ ,
\]
defines a meromorphic Siegel Modular Form of weight
\[
k= \frac1{2} \sum_{\substack{f/e\in \mathcal{P}\\ \ell \in \mathbb{Z}}} \frac{h_e}{N_e} c_{f/e}(0,0)
\]
with respect to $\mathbf{\Gamma}_t(N)^+$ possibly with character. The character is determined by the zeroth Fourier-Jacobi coefficient of $B_\psi(\mathbf{Z})$, i.e., the coefficient of $s^C$ in the Fourier expansion, which is a Jacobi form of weight $k$ and index $C$ of the Jacobi subgroup of $\mathbf{\Gamma}_t(N)^+$.
\end{theorem}
\noindent Here $\mathcal{P}$ are the set of cusps for $\Gamma_0(N)$, $h_e$ is the width of the cusp at $f/e\in \mathcal{P}$ and $N_e=N/e$. For any congruence subgroup $\Gamma$, the cusps are defined as the $\Gamma$-equivalent classes of $\mathbb{Q}\cup \{\infty\}$ where $\mathbb{Q}$ denotes the set of rational numbers. The width of the cusp $\alpha \in \mathcal{P}$ is the minimal $h_e$ such that $(\begin{smallmatrix} 1 & h_e\\0 &1\end{smallmatrix}) \in \gamma^{-1} \Gamma \gamma$, where  $\gamma\in SL(2,\mathbb{Z})$ such that $\gamma(\infty) = \alpha$.

\noindent \textbf{Remark:} As discussed in \cite{Govindarajan:2018ted},
the condition  on $c_{f/e}(n,\ell)$ for $4nt -\ell^2\le0$ that $\frac{h_e}{N_e} c_{f/e}(n,\ell)\in \mathbb{Z}$ 
will be relaxed as follows. For all cusps that have identical values of $(N_e,h_e)$, we require that the sum of $\frac{h_e}{N_e} c_{f/e}(n,\ell)$ (with $4nt-\ell^2\le0$) for all such cusps be integral. This ensures  that $B_\psi$ has only zeros or poles at all divisors -- this is the reason for the condition in the theorem.

Define the projection\cite{Raum:2012}, $\pi_{FE}$, as follows
\begin{equation}
\pi_{FE} \left(\phi |M_{f/e}\right) (\tau,z) := \frac1{h_e} \sum_{b=0}^{h_e-1} \phi |M_{f/e} (\tau +b,z)\ ,
\end{equation}
where $M_{f/e} = \left(\begin{smallmatrix} f & * \\ e & *\end{smallmatrix}\right)\in SL(2,\mathbb{Z} )$ maps the cusp at $i\infty$ to $f/e$. It is the Fourier coefficients of the projected Jacobi form defined above that appears in the product formula.

Raum proved the modularity of the products such as those given in Eq. \eqref{productformulamain} by considering products of rescaled Borcherds products\cite{Raum:2012}. His results were however restricted to only conjugacy classes associated with elements of $M_{24}$ that were primes or powers of primes. Further, there were some computational errors in his work that we fix as well.  We not only extend his results but also provide a systematic method of obtaining the precise rescaled Borcherds products that are needed for all conjugacy classes. For all type-I conjugacy classes of $M_{24}$, we find that the product formula is equivalent to a single Borcherds formula. More generally, we find that the following holds.
\begin{prop} \label{modularpropM24}
 For $g\in M_{24}$, let $\rho_m=[g^m]$, $\rho=\rho_1$, $Z_{0,1}^{\rho_m}(\tau,z)$ be the Jacobi form of weight zero and index one corresponding to $g^m$ and $\Phi_k^{\rho}(\mathbf{Z})$ be the modular form defined by the multiplicative lift given in Eq. \eqref{productformulamain}. Then
\begin{equation*}
\Big(\Phi_k^{\rho}(\mathbf{Z})\Big)^{p_{\rho}} = \prod_{m|p_\rho} \left(B_{\frac{p_\rho}m Z^{\rho_m}}(m\,\mathbf{Z})\right) \ ,
\end{equation*}
where $p_{\rho}$ is the length of the shortest cycle in the cycle shape for the conjugacy class $\rho$ and $B_\psi(\mathbf{Z})$ is the Siegel Modular Form given by Theorem \ref{CGproduct}.
\end{prop}

\begin{proof}
As the proof is mostly done by exhaustion, we illustrate this theorem for two examples, one with $p_\rho=1$ when $N$ is prime and one with $p_\rho\ne1$ here. All other conjugacy classes are worked out in the Appendix \ref{proofModularity}. 

\

\noindent $\boxed{\mathbf{p_{\rho}=1}}$\\

\noindent For all type-I conjugacy classes of $M_{24}$ i.e with $p_{\rho}=1$, we see that the multiplicative lift can be expressed in terms single Borcherds product. It is very easy to see for all the conjugacy classes with $N=1,2,3,5,7,11,23$ where $N$ is the order of the element $g\in M_{24}$. For prime $N$, there are only two cusps, one is at $i\infty$ (which is $\Gamma_0(N)$ equivalent to $1/N$) and another is at $0/1$. To prove the equality of multiplicative lift and the Borcherds product, we need to show that
\begin{align}
\pi_{FE} (Z^\rho) &= Z^\rho\ ,\\
\pi_{FE} (Z^\rho |S) &= \frac1N \left( Z^{1^{24}}- Z^\rho\right)\ . \label{check2}
\end{align}
The first equation holds trivially since $Z^{\rho} $ has only integral powers of $q$ in its Fourier-Jacobi expansion. The second part follows from a calculation.
\[
\pi_{FE}(Z^\rho|S)(\tau,z) = \frac{2}{N+1} \phi_{0,1} (\tau,z) + \pi_{FE}(\alpha^{(N)}|S)(\tau) \phi_{-2,1}(\tau,z)\ ,
\]
where $\alpha^{(N)}(\tau)=\frac{2N}{N+1} E_2^{(N)}(\tau)$ for $N=2,3,5,7$.
For $N=11$ and $23$, one has
\begin{align*}
\alpha^{(11)}(\tau)&=\frac{11}{6} E_2^{(11)}(\tau)-\frac{22}5 \eta_{1^211^2}(\tau)\ , \\
\alpha^{(23)}(\tau) &= \tfrac{23}{12} E_2^{(23)}(\tau) - \tfrac{23}{22} f_{23,1}(\tau) -\tfrac{161}{22} f_{23,2}(\tau) \ ,
\end{align*}
where $f_{23,a}(\tau)$ are defined in Appendix \ref{a1}.
For $N=2,3,5,7$, computing $\pi_{FE}\left(\alpha^{(N)}|S\right)(\tau)$, we obtain
\[
\pi_{FE}(\alpha^{(N)}|S)(\tau) = -\tfrac{2}{N(N+1)} \sum_{b=0}^{N-1} E_2^{(N)}(\tfrac{\tau+b}N)= -\tfrac{2}{(N+1)} E_2^{(N)}\big|_2U_N = -\tfrac{2}{(N+1)} E_2^{(N)}(\tau)\ ,
\]
where $U_N$ is the Hecke operator for $\Gamma_0(N)$ (defined by Atkin and Lehner~\cite{Atkin:1970}) and $E_2^{(N)}$ is its eigenform with eigenvalue $+1$.
Thus, we get
\begin{align*}
\pi_{FE}(Z^\rho|S)(\tau,z) &= \frac{2}{N+1} \phi_{0,1} (\tau,z) -\frac1N \alpha^{(N)}(\tau)\phi_{-2,1}(\tau,z)\\
&= \frac{2}{N} \phi_{0,1} (\tau,z) -\frac1N Z^\rho(\tau,z)\\
&= \frac1N \left(Z^{1^{24}}-Z^\rho\right)(\tau,z)\ ,
\end{align*}
which establishes Eq. \eqref{check2} for prime $N=2,3,5,7$. A similar computation holds for $N=11$ and $23$. For $N=11$,  one can show that $$\pi_{FE}(\alpha^{(11)}|S)(\tau)=-(1/11)\alpha^{(11)}(\tau)\ ,$$
and for $N=23$, one can show that
$$\pi_{FE}(\alpha^{(23)}|S)(\tau)=-(1/23)\alpha^{(23)}(\tau)\ .$$
Thus Eq. \eqref{check2} holds for $N=11$ and $23$ as well.\\ 

\noindent $\boxed{\mathbf{p_{\rho}\ne1}}$\\

\noindent Let us consider the cycle shape $12^2$. The multiplicative lift leads to the product formula given by (data about the cusps are given in \ref{Details12squared})
\begin{multline}\label{multiplicative12b}
\Phi^{12^2}_{-1} (\mathbf{Z}) = qrs\prod_{(n,\ell,m)>0} \Big(1-q^n r^\ell s^m \Big)^{f_1(nm,\ell)} \Big(1-(q^n r^\ell s^m)^2 \Big)^{x_1} \Big(1-(q^n r^\ell s^m)^3 \Big)^{x_2} \\ \Big(1-(q^n r^\ell s^m)^4\Big)^{x_3}
\Big(1-(q^n r^\ell s^m)^6 \Big)^{x_4} \Big(1-(q^n r^\ell s^m)^{12}\Big)^{x_5}\ ,
\end{multline}
where
\begin{align*}
&x_1=\frac12(f_2(nm,\ell)-f_1(nm,\ell)), \hspace{0.5cm} x_2=\frac13(f_3(nm,\ell)-f_1(nm,\ell)),\\
& x_3= \frac14 (- f_2(nm,\ell) + f_4(nm,\ell)),\\
&x_4=\frac13(f_1(nm,\ell)-f_2(nm,\ell)-f_3(nm,\ell)+f_6(nm,\ell)),\\
&x_5=\frac1{12} (f_0(nm,\ell) +f_2(nm,\ell) - f_4(nm,\ell) - f_6(nm,\ell)) \ .
\end{align*}
Now to see if this can be written as the Borcherds product as stated in Theorem \ref{CGproduct} we need to compute the coefficients at different cusps.
We need to combine cusps with same $(h_e, N_e)$ value and one can show that the \textit{only} non-zero contributions are the following:
\begin{align*}
\pi_{FE} \Big( Z^{12^2}|M_{\frac1{72}} \Big) &= - Z^{12^2} ,\\
\pi_{FE} \Big( Z^{ 12^2}|M_{\frac5{48}} +Z^{ 12^2}|M_{\frac1{48}} \Big) &= - Z^{12^2} ,\\
\pi_{FE} \Big( Z^{ 12^2}|M_{\frac5{24}} +Z^{ 12^2}|M_{\frac1{24}} \Big) &= Z^{12^2} \ .
\end{align*}
Using these and the product formula for Borcherds product given by Cl\'ery- Gritsenko, we get
\begin{multline}
B_{Z^{ 12^2}}(\mathbf{Z})= \prod_{(n,\ell,m)>0} \Big(1- q^n r^\ell s^m\Big)^{f_1(nm,\ell)} \Big(1- (q^n r^\ell s^m)^2\Big)^{-\frac12 f_1(nm,\ell)} \\
\Big(1- (q^n r^\ell s^m)^3\Big)^{-\frac13 f_1(nm,\ell) } \Big(1- (q^n r^\ell s^m)^6\Big)^{\frac16 f_1(nm,\ell)} \ .
\end{multline}
This does not give all the terms that appear in the product form given in Eq. \eqref{multiplicative12b}. The missing terms can be accounted for by additional terms leading to
\begin{align}
\Big( \Phi_{-1}^{ 12^2}(\mathbf{Z})\Big)^{12} = B_{12Z^{ 12^2}}(\mathbf{Z}) B_{6Z^{ 6^4}}(2\mathbf{Z}) B_{4Z^{4^6}}(3\mathbf{Z}) B_{3Z^{3^8}}(4\mathbf{Z}) B_{2Z^{2^{12}}}(6\mathbf{Z}) B_{Z^{1^{24}}}(12\mathbf{Z})\ .
\end{align}
$\Big( \Phi_{-1}^{ 12^2}(\mathbf{Z})\Big)^{12} $ is a meromorphic Siegel Modular Form of weight $(-12)$ at level 24. The zeroth Fourier-Jacobi term is given by the additive seed $\phi_{-1,1}(\tau,z) = \frac{\theta_1(\tau,z)^2}{\eta(\tau)^6} \eta(12\tau)^2$.
\end{proof}

\section{Conclusions}
We proved that there is a genus-two Siegel Modular Form (of a level $N$ subgroup of $Sp(2,\BZ)$) for every conjugacy class of $M_{24}$. Each of these arise on evaluating the class function for $M_{24}$ in Eq. \eqref{multiplicativelift} on the conjugacy class. This was achieved by expressing a product formula implied by Mathieu moonshine as rescaled product of the Borcherds products for all conjugacy classes. 
This  makes their modularity manifest. This completes the proof that was initiated by Raum\cite{Raum:2012} as well as resolving puzzles raised in \cite{Govindarajan:2011em}. Our result also corrects Raum's incorrect results for the cycle shapes $2^{12}, 3^8, 2^4 4^4, 4^6$, and $1^2 2^1 4^1 8^2$.

Generalized Mathieu moonshine deals with pairs of commuting elements of $(g,h)\in M_{24}$. To each pair, we have an eta product as well as a Jacobi form\cite{Gaberdiel:2012gf}.  Again, one has a product formula implied by second-quantized moonshine\cite{Govindarajan:2011em,Persson:2013xpa}. Their modularity has been proved in \cite{Persson:2013xpa}. Our methods can be used to provide  explicit formulae similar to the ones obtained in this paper.

Umbral moonshine is a generalization of Mathieu moonshine\cite{Cheng:2012tq,Cheng:2013wca}. These connect finite groups, vector valued mock modular forms to the Niemeier lattices in 24 dimensions. In a recent paper~\cite{Govindarajan:2019ezd}, we have shown that there exist Siegel Modular Forms for the $A$-series of umbral moonshine. It is of interest if such Siegel Modular Forms appear for all other cases.

It has been shown that  there is no point in the space of CFTs associated with the K3 sigma model where $M_{24}$ appears as a symmetry\cite{Gaberdiel:2011fg}. The best known example obtains $\mathbb{Z}^8_2$ : $\mathbb{M}_{20}$ symmetry\cite{Gaberdiel:2013psa}. There has been an interesting attempt to combine symmetries  at different points in the CFT moduli space and obtain all of $M_{24}$\cite{Taormina:2011rr}. It remains to be seen if there is a natural construction of a $M_{24}$ modules that leads to the class functions associated with the Jacobi and Siegel Modular Forms. \\

\noindent \textbf{Acknowledgments:} We thank Mohammed Shabbir and S. Viwanath for useful discussions. SS was supported by the Ramanujan Fellowship of Ayan Mukhopadhyay (IIT Madras) when some of this work was carried.

\appendix

\section{The modular group and congruence subgroups}

\textbf{Modular group: } It is a group of all fractional linear transformations on the complex upper half plane, i.e, the transformations $z\mapsto \frac{az+b}{cz+d}$, such that $a,b,c,$ and $d$ are all integers and obey $ad-bc=1$. 

These transformations are isomorphic to the group $PSL(2,\mathbb{Z})=SL(2,\mathbb{Z})/\{\pm\mathbf{1}\}$. The projective special linear group which we denote by $\Gamma$, is defined as follows:
$$\Gamma=\Big\{\left(\begin{smallmatrix} a&b \\c&d\end{smallmatrix}\right)| a,b,c,d\in \mathbb{Z}, ad-bc=1 \Big\}\ .$$
The complex upper half plane is defined as
$$\mathbb{H}_1=\{z\in\mathbb{C}; \text{Im } (z)>0\}\ .$$
The modular group is generated by following two transformation:
\begin{align*}
&T=\left(\begin{smallmatrix} 1&1\\0&1\end{smallmatrix}\right)\, \text{ such that } T\cdot z = z+1\\
&S=\left(\begin{smallmatrix} 1&-1\\1&0\end{smallmatrix}\right)\, \text{ such that } S\cdot z = -\frac1z
\end{align*}
Now we define various subgroups of $SL(2,\mathbb{Z})$ which will be useful for later discussions. These subgroups are obtained by imposing certain constrains on the matrix elements and are called congruence subgroups.
For an integer $N>1$, the {\it principal congruence subgroup} of level $N$ is defined as follows:
$$ \Gamma(N)=\Big\{ \left(\begin{smallmatrix} a&b\\c&d\end{smallmatrix}\right)\in SL(2,\mathbb{Z})| \left(\begin{smallmatrix} a&b\\c&d\end{smallmatrix}\right) \equiv \left(\begin{smallmatrix} 1&0\\0&1\end{smallmatrix}\right)\mod N\Big\} \ .$$
Another subgroup known as {\it Hecke congruence subgroup} of level $N$ denoted by $\Gamma_0(N)$, is defined as 
 $$ \Gamma_0(N)=\Big\{\left( \begin{smallmatrix} a&b\\c&d\end{smallmatrix}\right)\in SL(2,\mathbb{Z})| \left(\begin{smallmatrix} a&b\\c&d\end{smallmatrix}\right)\equiv \left(\begin{smallmatrix} *&*\\0&*\end{smallmatrix}\right)\mod N\Big\} \ .$$

\noindent\textbf{Symplectic group: } The symplectic group, $Sp(2,\mathbb{Q})$ is a set of $4\times 4$ matrices written in terms of four $2\times 2 $ matrices $A,B,C,D$ (with $A,B,C,D$ taking values in $\mathbb{Q}$) as $M=\left(\begin{smallmatrix} A&B\\C&D\end{smallmatrix}\right)$ satisfying $AB^T =BA^T$, $CD^T = DC^T$ and $AD^T-BC^T=I$. This group naturally acts on the Siegel upper half space $\mathbb{H}_2$, as
\begin{align}
\mathbf{Z} =\left(\begin{smallmatrix} \tau & z\\ z&\tau'\end{smallmatrix}\right) \longmapsto M\cdot\mathbf{Z} \equiv (A\mathbf{Z}+B)(C\mathbf{Z}+D)^{-1}\ .
\end{align}
\noindent The Siegel upper half space is defined as follows:
\begin{align}
\mathbb{H}_2 =\Big\{\mathbf{Z} =\left(\begin{smallmatrix} \tau& z\\z&\tau'\end{smallmatrix}\right)\in M_2(\mathbb{C}),\  \text{Im}(\mathbf{Z})>0\Big\}\ .
\end{align}
The paramodular group at level $N$, $\mathbf{\Gamma}_t(N)$, is a subgroup of $Sp(2,\mathbb{Q})$ is defined as follows:
\begin{align}
\mathbf{\Gamma}_t(N)=\left\{\left(\begin{smallmatrix} * & *t & * & *\\ * &*& *& *t^{-1}\\*N & *Nt & * & *\\ * Nt & *Nt &*t &*\end{smallmatrix}\right)\in Sp(2,\mathbb{Q}), \text{ with } N,t\in\mathbb{Z}_{>0}\text{ and } *\in\mathbb{Z}\right\}\ .
\end{align}
The $t=1, N=1$ case of the general paramodular group $\mathbf{\Gamma}_t(N)$ is the usual symplectic group over integers,  i.e., $\mathbf{\Gamma}_1=Sp(2,\mathbb{Z})$. The embedding of 
$\gamma=\left(\begin{smallmatrix} a & b \\ c & d\end{smallmatrix}\right)
\in \Gamma_0(N)$ in $\Gamma_t(N)$ is given by
\begin{equation}
\label{sl2embed}
\widetilde{\gamma}=\widetilde{\begin{pmatrix} a & b \\ c & d \end{pmatrix}}
\equiv \begin{pmatrix}
   a   &  0 & b & 0   \\
     0 & 1 & 0 & 0 \\
     c &  0 & d & 0 \\
     0 & 0 & 0 & 1  
\end{pmatrix}
\ , \ c=0 \mod N \ .
\end{equation}

The normal double extension of $\mathbf{\Gamma}_t(N)$ in $Sp(2,\mathbb{R})$ is defined as $\mathbf{\Gamma}_t(N)^+=\mathbf{\Gamma}_t(N)\cup \mathbf{\Gamma}_t(N)V_t$ where 
\begin{align}
V_t = \frac1{\sqrt{t}}\left(\begin{smallmatrix}0& t & 0 & 0\\ 1 &0& 0& 0\\0 & 0 & 0 & 1\\ 0 &0 & t &0\end{smallmatrix}\right)\in Sp(2,\mathbb{R})\ .
\end{align}
As shown by Cl\'ery and Gritsenko \cite{Clery2008}, this $\mathbf{\Gamma}_t(N)^+$ is generated by $V_t$ and by its parabolic subgroup
\begin{align}
\mathbf{\Gamma}_t^\infty(N) =\left\{ \pm \left(\begin{smallmatrix} * & 0 & * & *\\ * &1& *& *t^{-1}\\*N & 0 & * & *\\ 0 & 0 &0 &1\end{smallmatrix}\right)\in \Gamma_t(N),\ \forall *\in\mathbb{Z} \right\}\ .
\end{align}

\section{Modular forms and Jacobi forms}
\textbf{Modular form:} A modular form of weight $k$, with $k$ being integer or half integer, and character $v$ with respect to any of the subgroups of the modular group $\Gamma$ is a holomorphic function $f:\mathbb{H}_1 \rightarrow \mathbb{C}$ satisfying,
$$( f|_k M)(z) = v(M) f(z)\ ,$$
where $M=\left(\begin{smallmatrix}a&b\\c&d\end{smallmatrix}\right)\in \Gamma$ and $z\in \mathbb{H}_1$ and the slash operation is given by
$$ ( f|_k M)(z) =(c\tau+d)^{-k} f(M\cdot z)\ .$$

\

\noindent\textbf{Jacobi form:} A holomorphic function $\phi_{k,m}(\tau,z): \mathbb{H}_1\times \mathbb{C}\rightarrow \mathbb{C}$ is called a Jacobi form of weight $k$ and index $m$ if the function
$$\widetilde{\phi}_k (\mathbf{Z}) = \exp(2\pi i m\tau')\phi_{k,m}(\tau,z)\ ,$$
on $\mathbb{H}_2$ such that $\widetilde{\phi}_k (\mathbf{Z}) $ is a modular form of weight $k$ with respect to the Jacobi group $\mathbf{\Gamma}^J(N)\subseteq Sp(2,\mathbb{Z})$ with character $v$, i.e., it satisfies
\begin{equation}
\widetilde\phi|_k\ M (\mathbf{Z}) = v(M)\ \widetilde\phi_k(\mathbf{Z})\hspace{1cm} \forall M\in \Gamma^J(N)\ ,
\end{equation}
and it is holomorphic at all cusps, i.e.,
let $\gamma\in SL(2,\mathbb{Z})$ and $\tilde{\gamma}$, its embedding in $Sp(2,\mathbb{Z})$ as given in Eq. \eqref{sl2embed}, then it has a Fourier expansion
\begin{equation}
\widetilde\phi|_{k} \tilde{\gamma} (\mathbf{Z}) = s^m\ \sum_{\substack{n,\ell\\ 4nm- \ell^2\ge0}}\ c_\gamma(n,\ell)\ q^n r^\ell 
\end{equation}
where $q= e^{2\pi i \tau}$, $r=e^{2\pi i z}$ and $s= e^{2\pi i \tau'}$ , $n,\ell\in\mathbb{Q}$. \\

\noindent\textbf{Siegel Modular Form:} A Siegel Modular Form of weight $k$ for the subgroup $\mathbf{\Gamma}_t(N)$ with a character $v$ is a holomorphic function $F:\mathbb{H}_2\rightarrow \mathbb{C}$ which satisfies
\begin{align}
(F|_k M)(\mathbf{Z}) = v(M) F(\mathbf{Z})\ ,
\end{align}
for all $M= \left(\begin{smallmatrix}A & B\\ C & D\end{smallmatrix}\right)\in \mathbf{\Gamma}_t(N)$.

\noindent Here $|_k$ is the standard slash operator on the space of functions on $\mathbb{H}_2:$
\begin{align}\label{slash}
(F|_k M)(\mathbf{Z}) := \det(C\mathbf{Z}+D)^{-k} F(M\cdot\mathbf{Z})\ .
\end{align}

\subsection{Examples of Modular and Jacobi forms}\label{a1}

We discuss some examples of modular and Jacobi forms that enter our discussions. The Dedekind eta function is defined by
\[
\eta(\tau) = q^{1/24} \prod_{m=1}^{\infty} (1-q^m)\ .
\]
Its modular properties are given in Eq. \eqref{etatransform}.  $E_2^*(\tau)$ is the weight two non-holomorphic modular form of $SL(2,\mathbb{Z})$. It is given by
\begin{equation}
E_{2}^*(\tau)=1-24\ \sum_{n=1}^\infty \sigma_{1}(n)\ q^n\   -\frac{3}{\pi\ \textrm{Im} \tau}\ ,
\end{equation}
where $\sigma_{\ell}(n)=\sum_{1\leq d|n} d^\ell$. Then for $N>1$,
\begin{equation}
E_2^{(N)}(\tau):=\frac1{N-1}\Big(N E_2^*(N\tau)-E_2^*(\tau)\Big)
\end{equation}
is a weight two holomorphic modular form of $\Gamma_0(N)$ with constant coefficient equal to $1$. For $\Gamma_0(23)$, we need the following  weight-two modular forms that we denote by $f_{23,1}(\tau)$ and 
$f_{23,2}(\tau)$. The first few terms in their $q$-series are\cite{Gaberdiel:2010ch}
\begin{align*}
f_{23,1}(\tau)&= 2q- q^2-q^4-2q^5-5q^6+2q^7+4q^9+6q^{10} -6q^{11}+\cdots\ ,\\
f_{23,2}(\tau)&= q^2-2q^3-q^4+2q^5+q^6+2q^7-2q^8-2q^{10} -2q^{11}+\cdots\ .
\end{align*}

For $a.b\in (0,1)\mod 2$, the genus-one theta functions are defined by
\begin{equation}
\theta\left[\genfrac{}{}{0pt}{}{a}{b}\right] \left(\tau,z\right)
=\sum_{l \in \BZ} 
q^{\frac12 (l+\frac{a}2)^2}\ 
r^{(l+\frac{a}2)}\ e^{i\pi lb}\ ,
\end{equation}
with $r=\exp(2\pi i z)$.
Let $\vartheta_1 
\left(\tau,z\right)\equiv\theta\left[\genfrac{}{}{0pt}{}{1}{1}\right] 
\left(\tau,z\right)$, $\vartheta_2 
\left(\tau,z\right)\equiv\theta\left[\genfrac{}{}{0pt}{}{1}{0}\right] 
\left(\tau,z\right)$, $\vartheta_3 
\left(\tau,z\right)\equiv\theta\left[\genfrac{}{}{0pt}{}{0}{0}\right] 
\left(\tau,z\right)$ and $\vartheta_4 
\left(\tau,z\right)\equiv\theta\left[\genfrac{}{}{0pt}{}{0}{1}\right] 
\left(\tau,z\right)$. Using these, we can construct two index one Jacobi forms:
\begin{align}
\phi_{0,1}(\tau,z) &= 4\sum_{a=3}^4 \left[\frac{\vartheta_a(\tau,z)}{\vartheta_a(\tau,0)}\right]^2\ ,\\
\phi_{-2,1}(\tau,z) &=\frac{\vartheta_1(\tau,z)^2}{\eta(\tau)^6}\ .
\end{align}
Any Jacobi form, $Z_{0,1}(\tau,z)$, of weight zero and index one of $\Gamma_0(N)$ ($N>1$) can be written as follows\cite{Aoki:2005}
\[
Z_{0,1}(\tau,z) = A\ \phi_{0,1}(\tau,z) + B(\tau)\ \phi_{-2,1}(\tau,z) \ ,
\]
where $A$ is a constant and $B(\tau)$ is a weight two modular form of $\Gamma_0(N)$.

\section{List of $M_{24}$ Jacobi forms}\label{EOT}
There are 26 conjugacy classes of $M_{24}$. The ATLAS nomenclature of the conjugacy classes are as follows:
\begin{align*}
&\text{Type I : } 1a, 2a, 3a, 5a, 4b, 7a, 7b, 8a, 6a, 11a, 15a, 15b, 14a, 14b, 23a, 23b,\\
&\text{Type II : } 12b, 6b, 4c, 3b, 2b, 10a, 21a, 21b, 4a, 12a.
\end{align*}
The cycle shapes associated with type I conjugacy classes have at least one one-cycle.
\begin{table}[ht!]
\setlength{\extrarowheight}{3pt}
\begin{center}
\footnotesize{
\begin{tabular}{c|c|p{10cm}}
\hline
Conj. class & $\rho$ & \hspace{2in}$Z^\rho(\tau,z)$ \\ \hline
1a & $1^{24}$ & 2 $\phi_{0,1}(\tau,z)$\phantom{\Big|} \\[2pt]
2a & $1^{8}2^8$ & $\frac23\phi_{0,1}(\tau,z)+\frac43E_2^{(2)}(\tau)\phi_{-2,1}(\tau,z)$ \\[2pt]
3a & $1^63^6$ & $\tfrac12 \phi_{0,1}(\tau,z) +\tfrac32 E^{(3)}_2(\tau)\ \phi_{-2,1}(\tau,z) $ \\[2pt]
5a & $1^45^4$ & $\tfrac13 \phi_{0,1}(\tau,z) +\tfrac53 E^{(5)}_2(\tau)\ \phi_{-2,1}(\tau,z) $ \\[2pt]
4b & $1^42^2 4^4$& $\tfrac13 \phi_{0,1}(\tau,z) +(-\tfrac13 E^{(2)}_2(\tau)+2 E^{(4)}_2(\tau))\ \phi_{-2,1}(\tau,z) $ \\[2pt]
7a/b & $1^37^3$& $\frac14\phi_{0,1}(\tau,z)+\frac74E_2^{(7)}(\tau) \phi_{-2,1}(\tau,z) $ \\[2pt]
8a &$1^2 2^1 4^1 8^2$ & $\tfrac16 \phi_{0,1}(\tau,z) +(-\tfrac12 E^{(4)}_2(\tau)+\tfrac73 E^{(8)}_2(\tau))\ \phi_{-2,1}(\tau,z) $ \\[2pt]
6a & $1^22^23^26^2$ & $\tfrac16 \phi_{0,1}(\tau,z) +(-\tfrac16 E^{(2)}_2(\tau)-\tfrac12 E^{(3)}_2(\tau)+\tfrac52 E^{(6)}_2(\tau))\ \phi_{-2,1}(\tau,z) $ \\[2pt]
11a &$1^2 11^2$ & $\tfrac16 \phi_{0,1}(\tau,z) + (\tfrac{11}6 E^{(11)}_2(\tau) - \tfrac{22}5 \eta_{1^211^2}(\tau))\ \phi_{-2,1}(\tau,z) $ \\[2pt]
15a/b &$1^1 3^1 5^1 15^1$ & $\tfrac1{12} \phi_{0,1}(\tau,z) + (-\tfrac1{16} E^{(3)}_2(\tau) - \tfrac5{24}E_2^{(5)}(\tau)\newline+\tfrac{35}{16}E_2^{(15)}(\tau)-\frac{15}4 \eta_{1^13^15^115^1}(\tau))\ \phi_{-2,1}(\tau,z) $ \\[2pt]
14a/b &$1^1 2^1 7^1 14^1$ & $\tfrac1{12} \phi_{0,1}(\tau,z) + (-\tfrac1{36} E^{(2)}_2(\tau) - \tfrac7{12}E_2^{(7)}(\tau)\newline+\tfrac{91}{36}E_2^{(14)}(\tau)-\frac{14}3 \eta_{1^12^17^114^1}(\tau))\ \phi_{-2,1}(\tau,z) $ \\[2pt]
23a/b &$1^1 23^1$ & $\tfrac1{12} \phi_{0,1}(\tau,z) + (\tfrac{23}{12} E^{(23)}_2(\tau) - \tfrac{23}{22} f_{23,1}(\tau)\newline-\tfrac{161}{22}f_{23,2}(\tau))\ \phi_{-2,1}(\tau,z) $ \\[3pt]
\hline
12b & $12^2$ & $2\,\eta_{ 1^4 2^{-1}4^16^{1}12^{-1}}(\tau)\ \phi_{-2,1}(\tau,z) $ \\[2pt]
6b & $6^4$ & $2\, \eta_{ 1^2 2^23^{2}6^{-2}}(\tau)\ \phi_{-2,1}(\tau,z) $ \\[2pt]
4c & $4^6$ & $2\, \eta_{ 1^4 2^2 4^{-2}}(\tau)\ \phi_{-2,1}(\tau,z) $ \\[2pt]
3b & $3^{8}$ & $2\, \eta_{1^6 3^{-2}}(\tau)\ \phi_{-2,1}(\tau,z) $ \\[3pt]
2b & $2^{12}$ & $2\, \eta_{1^8 2^{-4}}(\tau)\ \phi_{-2,1}(\tau,z) $ \\[3pt]
10a & $2^2 10^2$ & $2\,\eta_{ 1^3 2^1 5^1 10^{-1}}(\tau)\ \phi_{-2,1}(\tau,z) $ \\[2pt]
21a/b & $3^1 21^1$ & $\left(\frac73 \eta_{ 1^3 3^{-1} 7^3 21^{-1}}(\tau)- \frac13 \eta_{1^6 3^{-2}}(\tau)\right)\ \phi_{-2,1}(\tau,z) $ \\[2pt]
4a & $2^4 4^4$ & $2\,\eta_{ 2^{8}4^{-4}}(\tau)\ \phi_{-2,1}(\tau,z) $ \\[2pt]
12a & $2^14^16^1 12^1$ &2$\, \eta_{ 1^3 2^{-1}3^{-1}4^26^{3}12^{-2}}(\tau)\ \phi_{-2,1}(\tau,z) $\\[2pt] \hline
\end{tabular}
}
\end{center}
\caption{The Jacobi forms for all conjugacy classes of $M_{24}$.}\label{twining_generaM24}
\end{table}
\clearpage

\section{Rules for computing cusps }\label{rulesofcusps}

The $S$ and $T$-transformation of $\eta$-function with simple arguments are as follows (see Chapter 2 of \cite{Peter2000})
\begin{align}
&\eta(N\tau)|_S = \tfrac1{\sqrt{N}}e^{-i\pi/4} \eta\left(\tfrac{\tau}N\right), \nonumber \\
&\left.\eta\left(\frac{\tau}N\right)\right|_S = \sqrt{N} e^{-i\pi/4}\eta(N\tau), \label{etatransform}\\
&\eta(\tau+a) = e^{i\pi a/12} \eta(\tau)\ . \nonumber
\end{align}
For more complicated arguments of $\eta$-function, let $\psi_j := \eta\left(\frac{\tau+j}{N}+j\right)$. Consider the following cases
\begin{enumerate}
\item If $(j,N)=1$, then choose a $j'$ such that $jj'=N+1$ we have,
\begin{align}\label{s_transformation}
\psi_j|_S(\tau) = \text{sign}(-j') e^{-3\pi i/4} \psi_{-j'}(\tau)\ .
\end{align}

\item  $(j,N)\ne 1$, but a simple change in variable of $\tau$ leads to the condition $(j,N)=1 $. In that case, we make the necessary change in the variable.  Then we choose a $j'$ such that $jj'=N+1$. An example of such case is the following.
\begin{align*}
\left.\eta\left(\frac{\tau-3}{6}\right)\right|_S = \tau^{-1/2} \eta\left(\frac{-\frac1{\tau}-3}{6}\right) = \tau^{-1/2}\eta\left(\frac{-\frac1{3\tau}-1}{2}\right) = \tau^{-1/2} \eta\left(\frac{-1-\tilde\tau}{2\tilde\tau}\right)\ ,
\end{align*}
where $\tilde\tau = 3\tau$. Now let us consider $\psi^{(2)}_{-1}(\tilde\tau) = \eta\left(\frac{\tilde\tau-1}{2}-1\right)$ and do the $S$-transformation.
\begin{multline*}
\left.\eta\left(\frac{\tilde\tau-1}{2}-1\right)\right|_S = e^{-i\pi/12} \left.\eta\left(\frac{\tilde\tau-1}{2}\right)\right|_S = e^{-i\pi/12} \tilde\tau^{-1/2} \eta\left( \frac{-\frac{1}{\tilde\tau} -1 }{2} \right)\\ = e^{-i\pi/12} \tilde\tau^{-1/2} \eta\left(\frac{-1-\tilde\tau}{2\tilde\tau}\right)\ .
\end{multline*}
Again using Eq. \eqref{s_transformation} we get
\begin{align*}
\left.\eta\left(\frac{\tau-1}{2}-1\right)\right|_S = e^{-3\pi i/4} \eta\left(\frac{\tau+3}{2}+3\right) \ .
\end{align*}
Comparing these two and putting $\tilde\tau= 3\tau$ back we get
\begin{align*}
\left.\eta\left(\frac{\tau-3}{6}\right)\right|_S = \sqrt{3} e^{-i\pi/3} \eta\left(\frac{3\tau+1}{2}\right)\ .
\end{align*}

\item  $(j,N)\ne 1$, and there exists no change of variable in $\tau$ that can lead us to the condition $(j,N)=1$. In that case, we need to use a more general formula. Given $\psi_j(\tau)$, choose $j'$ such that $jj' = 1$ mod $N$ and
$$ G =\begin{pmatrix} j & \frac{jj'-1}{N}\\ N & j'\end{pmatrix} \in \Gamma_0(N)\ .$$
Let $G = S T^a S T^b S T^c S$, then
$$ \psi_j \Big| _{S}(\tau) = e^{-\pi i} e^{(j+j'+a+b+c)\pi i/12} \, i^J \, \psi_{-j'}(\tau)\ ,$$
where $J$ is the number of changes of sign in the sequence $1, c, bc-1, abc-a-c = N$.
\end{enumerate}

\section{Proving Proposition \ref{modularpropM24} }\label{proofModularity}

\subsection{ Type-I conjugacy classes}

\subsubsection{Cycle shape $\mathbf{1^4\cdot 2^2\cdot 4^4}$}

The twining genera $Z^{\rho_a}(\tau,z),\, a=0,\ldots 3 $, which contribute to the multiplicative lifts are the one corresponding to conjugacy class $4b$ and its various cusps.
\begin{align*}
&Z^{\rho_0} (\tau,z)= Z^{1^{24}}(\tau,z),\\
&Z^{\rho}(\tau,z)=Z^{\rho_3} (\tau,z)= Z^{1^4 2^2 4^4}(\tau,z),\\
&Z^{\rho_2}(\tau,z)= Z^{1^8 2^8}(\tau,z)\ .
\end{align*}
All these Jacobi forms are listed in table \ref{twining_generaM24}. The multiplicative seed is $Z^{1^4 2^2 4^4}(\tau,z)$, and this is a Jacobi form of $\Gamma_0(4)$. The Borcherds lift (multiplicative lift), given by Eq. \eqref{productformulamain} leads to the following product formula:
\begin{multline}\label{multiplicative4b1}
\Phi_3^{1^4 2^2 4^4} (\mathbf{Z}) = qrs\prod_{(n,\ell,m)>0} \left(1-q^n r^\ell s^m\right)^{c^0(nm,\ell)-c^2(nm,\ell)} \left(1-(q^n r^\ell s^m)^2\right)^{c^2(nm,\ell)-c^1(nm,\ell)}\\ \left(1-(q^n r^\ell s^m)^4\right)^{c^1(nm,\ell)} \ .
\end{multline}
Then using Eq. \eqref{inverseFJ} we get,
\begin{align}
c^0 = \frac14(f_0 +2 f_1+ f_2), \hspace{0.2cm} c^1 = \frac14(f_0 - f_2), \hspace{0.2cm} c^2 = \frac14(f_0 -2 f_1+ f_2)\ .
\end{align}
Using these we can rewrite the product formula \eqref{multiplicative4b1} as
\begin{multline}\label{multiplicative4b}
\Phi_3^{1^4 2^2 4^4} (\mathbf{Z}) = qrs\prod_{(n,\ell,m)>0} \left(1-q^n r^\ell s^m\right)^{f_1(nm,\ell)} \left(1-(q^n r^\ell s^m)^2\right)^{\frac12(f_2(nm,\ell)-f_1(nm,\ell))}\\
\left(1-(q^n r^\ell s^m)^4\right)^{\frac14(f_0(nm,\ell)-f_2(nm,\ell))} \ .
\end{multline}
To see if it can be written as the Borcherds product as stated in theorem \ref{CGproduct}, we need to compute the Fourier coefficients at different cusps. The cusps and other data for $\Gamma_0(4)$ are as follows:\footnote{Link to compute cusps of congruence subgroups $\Gamma_0(N)$: \url{http://magma.maths.usyd.edu.au/calc/} and the code is \\
$>$ G := CongruenceSubgroup(0,N);\\
$>$ Cusps(G);\\
$>$ Widths(G);} 
\begin{center}
\begin{tabular}{|c|c|c|c|}
\hline
$f/e$& $i\infty$ & $1/2$ & $0/1$\\
\hline
$h_e$ & 1&1& 4\\
\hline
$N_e$& 1& 2&4\\
\hline
\end{tabular}
\end{center}
Using $M_{i\infty} = 1, M_0 = S$, and $M_{1/p} = -ST^{-p}S$ and using the transformation rules under $S$ and $T$ as listed above in Section \ref{rulesofcusps}, we can show that
\begin{align*}
&\pi_{FE} \Big( Z^{1^4 2^2 4^4}|M_{\frac12} \Big) = Z^{1^8 2^8} - Z^{1^4 2^2 4^4},\\
&\pi_{FE} \Big( Z^{1^4 2^2 4^4}|M_{\frac01} \Big) = \frac14\left( Z^{1^{24}} - Z^{1^8 2^8}\right)\ .
\end{align*}
The Fourier expansion of the cusp about zero does not have integral power in $q$ and hence does not contribute to the product formula. Using these identities, we see that the Borcherds product as given by Clery-Gritsenko in \ref{CGproduct} takes the form same as Eq. \eqref{multiplicative4b}, i.e.,
\begin{align}
B_{Z^{1^4 2^2 4^4}}(\mathbf{Z})= \Phi_3^{1^4 2^2 4^4}(\mathbf{Z})\ .
\end{align}
The zeroth coefficient of the Siegel Modular Form, i.e., the $m=0$ term is given by the additive seed $\phi_{3,1}(\tau,z)= \frac{\theta_1(\tau,z)^2}{\eta(\tau)^6} \eta(\tau)^4 \eta(2\tau)^2\eta(4\tau)^4$. Hence $\Phi_3^{1^4 2^2 4^4}(\mathbf{Z})$ is a modular form of weight 3 at level 4.

\subsubsection{Cycle shape $\mathbf{1^2\cdot 2^2\cdot 3^2\cdot 6^2}$}

This is an element of order $6$. The twining genera that contribute to the multiplicative lifts are the following:
\begin{align*}
&Z^{\rho_0}(\tau,z)= Z^{1^{24}}(\tau,z),\\
&Z^{\rho} (\tau,z)=Z^{\rho_5} (\tau,z)=Z^{1^2 2^2 3^2 6^2}(\tau,z),\\
&Z^{\rho_2} (\tau,z)=Z^{\rho_4} (\tau,z)= Z^{1^6 3^6}(\tau,z),\\
&Z^{\rho_3}(\tau,z)= Z^{1^8 2^8}(\tau,z)\ .
\end{align*}
The multiplicative seed $Z^{1^2 2^2 3^2 6^2}(\tau,z)$ is a Jacobi form of $\Gamma_0(6)$. The cusps and other data for $\Gamma_0(6)$ are the following:
\begin{center}
\begin{tabular}{|c|c|c|c|c|}
\hline
$f/e$& $i\infty$ & $1/2$&$1/3$ & $0/1$\\
\hline
$h_e$ & 1&3&2& 6\\
\hline
$N_e$& 1&3&2& 6\\
\hline
\end{tabular}
\end{center}
The multiplicative lift given by Eq. \eqref{productformulamain} leads to the product formula,
\begin{multline}
\Phi_2^{1^2 2^2 3^2 6^2} (\mathbf{Z}) = qrs\prod_{(n,\ell,m)>0} \left(1-q^n r^\ell s^m\right)^{c^0(nm,\ell)+c^1(nm,\ell)-c^2(nm,\ell)-c^3(nm,\ell)} \\
\left(1-(q^n r^\ell s^m)^2\right)^{c^3(nm,\ell)-c^1(nm,\ell)}
\left(1-(q^n r^\ell s^m)^3\right)^{c^2(nm,\ell)-c^1(nm,\ell)} \left(1-(q^n r^\ell s^m)^6\right)^{c^1(nm,\ell)} \ .
\end{multline}
Then using Eq. \eqref{inverseFJ} we express the $c^\alpha$ in terms of the Fourier-Jacobi coefficients of $Z^{\rho_s}$ as follows:
\begin{align}
&c^0 = \frac16(f_0 +2 f_1+2 f_2+f_3), \hspace{0.2cm} c^1 = \frac16(f_0 +f_1- f_2-f_3),\nonumber\\
& c^2 = \frac16(f_0 - f_1-f_2+ f_3), \hspace{0.2cm} c^3 = \frac16(f_0 -2 f_1+ 2f_2 - f_3)\ .
\end{align}
Then the above product formula becomes:
\begin{multline}\label{multiplicative6a}
\Phi_2^{1^22^23^2 6^2} (\mathbf{Z}) = qrs\prod_{(n,\ell,m)>0} \left(1-q^n r^\ell s^m\right)^{f_1 (nm,\ell)} \left(1-(q^n r^\ell s^m)^2\right)^{\frac12(f_2(nm,\ell)-f_1(nm,\ell))}\\ \left(1-(q^n r^\ell s^m)^3\right)^{\frac13(f_3(nm,\ell)-f_1(nm,\ell))} \left(1-(q^n r^\ell s^m)^6\right)^{\frac16(f_0(nm,\ell)+f_1(nm,\ell)-f_2(nm,\ell)-f_3(nm,\ell))} \ .
\end{multline}
Now we need to check if it can be written as the Borcherds product as stated in theorem \ref{CGproduct}. For that we need to compute the Fourier-Jacobi coefficients at cusps $i\infty, 1/2, 1/3$ and $0/1$. Straightforward computation shows that
\begin{align*}
&\pi_{FE} \Big( Z^{1^2 2^2 3^2 6^2}|M_{\frac12} \Big) =\frac13\left( Z^{1^8 2^8} - Z^{1^2 2^2 3^2 6^2}\right),\\
&\pi_{FE} \Big( Z^{1^2 2^2 3^2 6^2}|M_{\frac13} \Big) =\frac12\left(Z^{1^6 3^6} - Z^{1^2 2^2 3^2 6^2} \right),\\
&\pi_{FE} \Big( Z^{1^2 2^2 3^2 6^2}|M_{\frac01} \Big) = \frac16\left( Z^{1^{24}} + Z^{1^2 2^2 3^2 6^2} -Z^{1^63^6} - Z^{1^8 2^8}\right)\ .
\end{align*}
Using these and using the product formula for Borcherds product given by Cl\'ery-Gritsenko, we get expression for $B_Z^{1^2 2^2 3^2 6^2}(\mathbf{Z})$ which is same as Eq. \eqref{multiplicative6a}, i.e.,
\begin{align}
B_{Z^{1^2 2^2 3^2 6^2}}(\mathbf{Z})= \Phi_2^{1^2 2^2 3^2 6^2}(\mathbf{Z})\ .
\end{align}
Hence $\Phi_2^{1^2 2^2 3^2 6^2}(\mathbf{Z})$ is a modular form of weight 2 and level 6. The zeroth coefficient of this Siegel Modular Form, i.e., the $m=0$ term is given by the additive seed $\phi_{2,1}(\tau,z) = \frac{\theta_1(\tau,z)^2}{\eta(\tau)^6} \eta(\tau)^2 \eta(2\tau)^2\eta(3\tau)^2\eta(6\tau)^2$.

\subsubsection{Cycle shape $\mathbf{1^2\cdot 2^1\cdot 4^1\cdot 8^2}$}

This is a $M_{24}$ element of order $8$ and the conjugacy class is $8a$. The twining genera that contribute to the multiplicative lift are the following:
\begin{align*}
&Z^{\rho_0} (\tau,z)= Z^{1^{24}}(\tau,z),\\
&Z^{\rho_1} (\tau,z)=Z^{\rho_3} (\tau,z)= Z^{\rho_5}(\tau,z)= Z^{\rho_7} (\tau,z)= Z^{1^2 2^1 4^1 8^2}(\tau,z),\\
&Z^{\rho_2} (\tau,z)=Z^{\rho_6} (\tau,z)= Z^{1^4 2^2 4^4}(\tau,z),\\
&Z^{\rho_4} (\tau,z)= Z^{1^8 2^8}(\tau,z)\ .
\end{align*}
The multiplicative seed is given by $Z^{1^2 2^1 4^1 8^2}$ and this is a Jacobi form of $\Gamma_0(8)$. The cusps at level $8$ are the following:
\begin{center}
\begin{tabular}{|c|c|c|c|c|}
\hline
$f/e$& $i\infty$ & $1/2$&$1/4$ & $0/1$\\
\hline
$h_e$ & 1&2&1& 8\\
\hline
$N_e$& 1& 4&2&8\\
\hline
\end{tabular}
\end{center}
The Borcherds lift given by Eq. \eqref{productformulamain} leads to the following product formula:
\begin{multline}
\Phi_1^{1^2 2^1 4^1 8^2} (\mathbf{Z}) = qrs\prod_{(n,\ell,m)>0} \left(1-q^n r^\ell s^m\right)^{c^0(nm,\ell)-c^4(nm,\ell)} \left(1-(q^n r^\ell s^m)^2\right)^{c^4(nm,\ell)-c^2(nm,\ell)} \\
\left(1-(q^n r^\ell s^m)^4\right)^{c^2(nm,\ell)-c^1(nm,\ell)} \left(1-(q^n r^\ell s^m)^8\right)^{c^1(nm,\ell)} \ .
\end{multline}
Then using Eq. \eqref{inverseFJ}, i.e.,
\begin{align}
&c^0 = \frac18(f_0 +4 f_1+2 f_2+f_4), \hspace{0.2cm} c^1 = \frac18(f_0 - f_4),\nonumber\\
& c^2 = \frac18(f_0 -2f_2+ f_4), \hspace{0.2cm} c^4 = \frac18(f_0 -4 f_1+ 2f_2 +f_4)\ ,
\end{align}
we rewrite the product formula in terms of the Fourier-Jacobi coefficients $f_a(n,\ell)$ as
\begin{multline}\label{multiplicative8a}
\Phi_1^{1^2 2^1 4^1 8^2} (\mathbf{Z}) = qrs\prod_{(n,\ell,m)>0} \left(1-q^n r^\ell s^m\right)^{f_1 (nm,\ell)} \left(1-(q^n r^\ell s^m)^2\right)^{\frac12(f_2(nm,\ell)-f_1(nm,\ell))}\\ \left(1-(q^n r^\ell s^m)^4\right)^{\frac14(f_4(nm,\ell)-f_2(nm,\ell))} \left(1-(q^n r^\ell s^m)^8\right)^{\frac18\left(f_0(nm,\ell)-f_4(nm,\ell)\right)} \ .
\end{multline}
To check if it can be written as the Borcherds product of Cl\'ery-Gritsenko \ref{CGproduct}, we need to compute the coefficients at cusps about $i\infty, 1/2,1/4$ and $0/1$. Computing $Z^{1^2 2^1 4^1 8^2}$ at these cusps one can show that
\begin{align*}
&\pi_{FE} \Big( Z^{1^2 2^1 4^1 8^2}|M_{\frac12} \Big) =\frac12\left( Z^{1^8 2^8} - Z^{1^4 2^2 4^4}\right),\\
&\pi_{FE} \Big( Z^{1^2 2^1 4^1 8^2}|M_{\frac14} \Big) =\left(Z^{1^4 2^2 4^4} - Z^{1^2 2^1 4^1 8^2} \right),\\
&\pi_{FE} \Big( Z^{1^2 2^1 4^1 8^2}|M_{\frac01} \Big) = \frac18\left( Z^{1^{24}} - Z^{1^8 2^8}\right)\ .
\end{align*}
Using these and using the product formula for Borcherds product given by Cl\'ery-Gritsenko, we get expression for $B_{Z^{1^2 2^1 4^1 8^2}}(\mathbf{Z})$ same as Eq. \eqref{multiplicative8a}, i.e.,
\begin{align}
B_{Z^{1^2 2^1 4^1 8^2}}(\mathbf{Z})= \Phi_1^{1^2 2^1 4^1 8^2}(\mathbf{Z})\ .
\end{align}
Hence $ \Phi_1^{1^2 2^1 4^1 8^2}(\mathbf{Z})$ is a meromorphic Siegel Modular Form of weight 1 at level 8. The zeroth coefficient of this Siegel Modular Form, i.e., the $m=0$ term is given by the additive seed $\phi_{1,1}(\tau,z)=\frac{\theta_1(\tau,z)^2}{\eta(\tau)^6} \eta(\tau)^2 \eta(2\tau)\eta(4\tau)\eta(8\tau)^2$.

\subsubsection{Cycle shape $\mathbf{1\cdot 2\cdot 7\cdot 14}$}

This is an element of order $14$. The following twining genera that contribute to the multiplicative lift :
\begin{align*}
&Z^{\rho_0}(\tau,z)= Z^{1^{24}}(\tau,z),\\
&Z^{\rho} (\tau,z)=Z^{\rho_3}(\tau,z)=Z^{\rho_5} (\tau,z)= Z^{\rho_9}(\tau,z)= Z^{\rho_{11}} (\tau,z)= Z^{\rho_{13}}(\tau,z)= Z^{1^1 2^1 7^1 14^1}(\tau,z),\\
&Z^{\rho_{2}} (\tau,z)= Z^{\rho_4} (\tau,z)= Z^{\rho_6}(\tau,z)= Z^{\rho_8} (\tau,z)= Z^{\rho_{10}} (\tau,z)= Z^{\rho_{12}} (\tau,z)= Z^{1^3 7^3}(\tau,z),\\
&Z^{\rho_7}(\tau,z)=Z^{1^8 2^8}(\tau,z)\ .
\end{align*}
The multiplicative seed for this case, i.e., $Z^{1^1 2^1 7^1 14^1}(\tau,z)$ is a Jacobi form of $\Gamma_0(14)$. The cusps and other data at level $14$ are the following:
\begin{center}
\begin{tabular}{|c|c|c|c|c|}
\hline
$f/e$& $i\infty$ & $1/2$&$1/7$ & $0/1$\\
\hline
$h_e$ & 1&7&2& 14\\
\hline
$N_e$& 1& 7&2&14\\
\hline
\end{tabular}
\end{center}
The Borcherds lift given by Eq. \eqref{productformulamain} leads to the following product formula:
\begin{multline}
\Phi_0^{1^1 2^1 7^1 14^1} (\mathbf{Z}) = qrs\prod_{(n,\ell,m)>0} \left(1-q^n r^\ell s^m\right)^{c^0(nm,\ell)+c^1(nm,\ell)-c^2(nm,\ell)-c^7(nm,\ell)}\\
\left(1-(q^n r^\ell s^m)^2\right)^{c^7(nm,\ell)-c^1(nm,\ell)} \left(1-(q^n r^\ell s^m)^7\right)^{c^2(nm,\ell)-c^1(nm,\ell)} \left(1-(q^n r^\ell s^m)^{14}\right)^{c^1(nm,\ell)} \ .
\end{multline}
Again we rewrite $c^{\alpha}(n,\ell)$ in terms of $f_a(n,\ell)$ using Eq. \eqref{inverseFJ},
\begin{align}
&c^0 = \frac1{14}(f_0 +6 f_1+6 f_2+f_7), \hspace{0.2cm} c^1 = \frac1{14}(f_0 +f_1-f_2- f_7),\nonumber\\
& c^2 = \frac1{14}(f_0 -f_1- f_2+ f_7), \hspace{0.2cm} c^7 = \frac1{14}(f_0 -6 f_1+ 6f_2 -f_7)\ .
\end{align}
Then the above product formula becomes
\begin{multline}\label{multiplicative14ab}
\Phi_0^{1^1 2^1 7^1 14^1} (\mathbf{Z}) = qrs\prod_{(n,\ell,m)>0} \left(1-q^n r^\ell s^m\right)^{f_1 (nm,\ell)} \left(1-(q^n r^\ell s^m)^2\right)^{\frac12(f_2(nm,\ell)-f_1(nm,\ell))}\\ \left(1-(q^n r^\ell s^m)^7\right)^{\frac17(f_7(nm,\ell)-f_1(nm,\ell))} \left(1-(q^n r^\ell s^m)^{14}\right)^{\frac1{14}\left(f_0(nm,\ell)+f_1(nm,\ell)-f_2(nm,\ell)-f_7(nm,\ell)\right)} \ .
\end{multline}
To see if it can be written as the Borcherds product as stated in theorem \ref{CGproduct} we need to compute the Fourier-Jacobi coefficients of the multiplicative seed at cusps about $i\infty, 1/2,1/7$ and $0/1$. Computing the cusps, one can show that
\begin{align*}
&\pi_{FE} \Big( Z^{1^1 2^1 7^1 14^1}|M_{\frac12} \Big) =\frac12\left( Z^{1^3 7^3} - Z^{1^1 2^1 7^1 14^1}\right),\\
&\pi_{FE} \Big( Z^{1^1 2^1 7^1 14^1}|M_{\frac17} \Big) =\frac17\left( Z^{1^8 2^8} - Z^{1^1 2^1 7^1 14^1} \right),\\
&\pi_{FE} \Big( Z^{1^1 2^1 7^1 14^1}|M_{\frac01} \Big) = \frac1{14}\left( Z^{1^{24}}+ Z^{1^1 2^1 7^1 14^1} - Z^{1^8 2^8}- Z^{1^3 7^3} \right)\ .
\end{align*}
Then the Borcherds product formula given by Cl\'ery-Gritsenko leads to product formula Eq. \eqref{multiplicative14ab}, i.e.,
\begin{align}
B_{Z^{1^1 2^1 7^1 14^1}}(\mathbf{Z})= \Phi_0^{1^1 2^1 7^1 14^1}(\mathbf{Z})\ .
\end{align}
Hence $ \Phi_0^{1^1 2^1 7^1 14^1}(\mathbf{Z})$ is a modular form of weight zero at level 14. The zeroth coefficient of the Siegel Modular Form, i.e., the $m=0$ term is given by the additive seed $\phi_{0,1}(\tau,z)= \frac{\theta_1(\tau,z)^2}{\eta(\tau)^6} \eta(\tau) \eta(2\tau)\eta(7\tau)\eta(14\tau)$.

\subsubsection{Cycle shape $\mathbf{1\cdot 3\cdot 5\cdot 15}$}

This is an $M_{24}$ element of order $15$. The twining genera that contribute to the multiplicative lift are the following:
\begin{align*}
&Z^{\rho_0}(\tau,z)= Z^{1^{24}}(\tau,z),\\
&Z^{\rho} (\tau,z)=Z^{\rho_2} (\tau,z)= Z^{\rho_4}(\tau,z)= Z^{\rho_7} (\tau,z),\\
&\hspace{0.5cm}=Z^{\rho_8} (\tau,z)= Z^{\rho_{11}} (\tau,z)= Z^{\rho_{13}} (\tau,z)= Z^{\rho_{14}} (\tau,z)= Z^{1^1 3^1 5^1 15^1}(\tau,z),\\
&Z^{\rho_3} (\tau,z)= Z^{\rho_6} (\tau,z)= Z^{\rho_9} (\tau,z)= Z^{\rho_{12}}(\tau,z)= Z^{1^4 5^4}(\tau,z),\\
&Z^{\rho_5} (\tau,z)= Z^{\rho_{10}} (\tau,z)= Z^{1^6 3^6}(\tau,z)\ .
\end{align*}
The multiplicative seed is a Jacobi form of $\Gamma_0(15)$. The cusps and other data for $\Gamma_0(15)$ are the following:
\begin{center}
\begin{tabular}{|c|c|c|c|c|}
\hline
$f/e$& $i\infty$ & $1/3$&$1/5$ & $0/1$\\
\hline
$h_e$ & 1&5&3& 15\\
\hline
$N_e$& 1& 5&3&15\\
\hline
\end{tabular}
\end{center}
The Borcherds lift given by Eq. \eqref{productformulamain} leads to the following product formula:
\begin{multline}
\Phi_0^{1^1 3^1 5^1 15^1} (\mathbf{Z}) = qrs\prod_{(n,\ell,m)>0} \left(1-q^n r^\ell s^m\right)^{c^0(nm,\ell)+c^1(nm,\ell)-c^3(nm,\ell)-c^5(nm,\ell)}\\
\left(1-(q^n r^\ell s^m)^3\right)^{c^5(nm,\ell)-c^1(nm,\ell)} \left(1-(q^n r^\ell s^m)^5\right)^{c^3(nm,\ell)-c^1(nm,\ell)} \left(1-(q^n r^\ell s^m)^{15}\right)^{c^1(nm,\ell)} \ .
\end{multline}
Then using Eq. \eqref{inverseFJ}, i.e.,
\begin{align}
&c^0 = \frac1{15}(f_0 +8 f_1+4 f_3+2f_5), \hspace{0.2cm} c^1 = \frac1{15}(f_0 +f_1-f_3- f_5),\nonumber\\
& c^3 = \frac1{15}(f_0 - 2f_1- f_3+2 f_5), \hspace{0.2cm} c^5 = \frac1{15}(f_0 - 4 f_1+ 4f_3 -f_5)\ ,
\end{align}
we can rewrite the above product formula as
\begin{multline}\label{multiplicative15ab}
\Phi_0^{1^1 3^1 5^1 15^1} (\mathbf{Z}) = qrs\prod_{(n,\ell,m)>0} \left(1-q^n r^\ell s^m\right)^{f_1 (nm,\ell)} \left(1-(q^n r^\ell s^m)^3\right)^{\frac13(f_3(nm,\ell)-f_1(nm,\ell))}\\ \left(1-(q^n r^\ell s^m)^5\right)^{\frac15(f_5(nm,\ell)-f_1(nm,\ell))} \left(1-(q^n r^\ell s^m)^{15}\right)^{\frac1{15}\left(f_0(nm,\ell)+f_1(nm,\ell)-f_3(nm,\ell)-f_5(nm,\ell)\right)} \ .
\end{multline}
To see if it can be written as the Borcherds product as stated in theorem \ref{CGproduct} we need to compute the Fourier coefficients at cusps about $i\infty, 1/3,1/5$ and $1/15$. The cusp about zero does not have any integer power in $q$ and hence does not contribute. Computing the seed, i.e., $Z^{1^1 3^1 5^1 15^1}(\tau,z)$ at these cusps one can show that
\begin{align*}
&\pi_{FE} \Big( Z^{1^1 3^1 5^1 15^1}|M_{\frac13} \Big) =\frac13\left( Z^{1^4 5^4} - Z^{1^1 3^1 5^1 15^1}\right),\\
&\pi_{FE} \Big( Z^{1^1 3^1 5^1 15^1}|M_{\frac15} \Big) =\frac15\left(Z^{1^6 3^6} - Z^{1^1 3^1 5^1 15^1} \right),\\
&\pi_{FE} \Big( Z^{1^1 3^1 5^1 15^1}|M_{\frac01} \Big) = \frac1{15}\left( Z^{1^{24}}+ Z^{1^1 3^1 5^1 15^1} - Z^{1^6 3^6}- Z^{1^4 5^4} \right)\ .
\end{align*}
Using these and using the product formula for Borcherds product given by Cl\'ery-Gritsenko, we get expression for $B_{Z^{1^1 3^1 5^1 15^1}}(\mathbf{Z})$ which is identical to
Eq. \eqref{multiplicative15ab}, i.e.,
\begin{align}
B_{Z^{1^1 3^1 5^1 15^1}}(\mathbf{Z})= \Phi_0^{1^1 3^1 5^1 15^1}(\mathbf{Z})\ .
\end{align}
Hence $ \Phi_0^{1^1 3^1 5^1 15^1}(\mathbf{Z})$ is a meromorphic Siegel Modular Form of weight zero at level 15. The zeroth coefficient of the Siegel Modular Form, i.e., the $m=0$ term is given by the additive seed $\phi_{0,1}(\tau,z) = \frac{\theta_1(\tau,z)^2}{\eta(\tau)^6} \eta(\tau) \eta(3\tau)\eta(5\tau)\eta(15\tau)$.

So we see that for all the type-I conjugacy classes of $M_{24}$, the length of the lowest cycle in the balanced cycle-shape is one and all the product formula obtained via the multiplicative lift matches with the Borcherds product formula given by Cl\'ery-Gritsenko.\\

\subsection{Type-II conjugacy classes}

\noindent Now we consider the cases for which there is no cycle shape of length one. The corresponding conjugacy classes are called type-II conjugacy classes.

\subsubsection{Cycle shape $\mathbf{2^{12}}$}

This is an element of order $2$. There are only two twining genera which contribute to the multiplicative lift. They are
\begin{align*}
&Z^{\rho_0}(\tau,z)= Z^{1^{24}}(\tau,z)\ ,\hspace{0.5cm} Z^{\rho} (\tau,z)= Z^{2^{12}}(\tau,z)\ .
\end{align*}
The multiplicative seed corresponding to this conjugacy class, i.e., $Z^{2^{12}}(\tau,z)$ is a Jacobi form of $\Gamma_0(4)$. The cusps and other data of $\Gamma_0(4)$ are the following:
\begin{center}
\begin{tabular}{|c|c|c|c|}
\hline
$f/e$& $i\infty$ & $1/2$ & $0/1$\\
\hline
$h_e$ & 1&1&4\\
\hline
$N_e$& 1& 2&4\\
\hline
\end{tabular}
\end{center}
The Borcherds lift given by Eq. \eqref{productformulamain} leads to the following product formula:
\begin{multline}
\Phi_4^{2^{12}} (\mathbf{Z}) = qrs\prod_{(n,\ell,m)>0} \left(1-q^n r^\ell s^m\right)^{c^0 (nm,\ell)- c^1(nm,\ell)} \left(1-(q^n r^\ell s^m)^2\right)^{c^1(nm,\ell)}\ .
\end{multline}
Then using Eq. \eqref{inverseFJ}, i.e.,
\begin{align}
&c^0 = \frac12(f_0+ f_1), \hspace{0.2cm} c^1 = \frac12(f_0 -f_1)\ ,
\end{align}
we can rewrite the above product formula as
\begin{align}\label{multiplicative2b}
\Phi_4^{2^{12}} = qrs \prod_{(n,\ell,m)>0} \Big( 1- q^n r^\ell s^m\Big)^{f_1(nm,\ell)} \Big( 1- (q^n r^\ell s^m)^2\Big)^{\frac12(f_0(nm.\ell)-f_1(nm,\ell))} \ .
\end{align}
To see if it can be written as the Borcherds product as stated in theorem \ref{CGproduct}, we need to compute the coefficients at the cusps $i\infty$ and $1/2$. The other cusps $0/1$ does not contribute as there are no integer power in $q$ in the Fourier expansion. Then we get
\begin{align*}
&\pi_{FE} \Big( Z^{2^{12}}|M_{\frac12} \Big) = - Z^{2^{12}} \ .
\end{align*}
Using this and the product formula for Borcherds product given by Cl\'ery- Gritsenko, we get
\begin{multline}
B_{Z^{2^{12}}}(\mathbf{Z})= \prod_{(n,\ell,m)>0} \Big(1- q^n r^\ell s^m\Big)^{f_1(nm,\ell)} \Big(1- (q^n r^\ell s^m)^2\Big)^{-\frac12 f_1(nm,\ell)} \ .
\end{multline}
This does not give all the terms that appear in the product form obtained from the Borcherds lift given in Eq. \eqref{multiplicative2b}. The missing terms can be accounted for by additional terms leading to
\begin{align}
\Big( \Phi_4^{2^{12}}(\mathbf{Z})\Big)^2 = B_{2Z^{2^{12}}}(\mathbf{Z}) B_{Z^{1^{24}}}(\mathbf{2Z})\ .
\end{align}
Therefore $\Big( \Phi_4^{2^{12}}(\mathbf{Z})\Big)^2 $ is a meromorphic Siegel Modular Form of weight $8$ at level $4$. The zeroth coefficient of the Siegel Modular Form, i.e., the $m=0$ term is given by the additive seed $\phi_{4,1}(\tau,z) = \frac{\theta_1(\tau,z)^2}{\eta(\tau)^6} \eta(2\tau)^{12}$.

\subsubsection{Cycle shape $\mathbf{2^4\cdot 4^4}$}

This $M_{24}$ element is of order $4$. The twining genera that contribute to the multiplicative lifts are,
\begin{align*}
&Z^{\rho_0} (\tau,z)= Z^{1^{24}}(\tau,z)\ ,\\
&Z^{\rho} (\tau,z)=Z^{\rho_3}(\tau,z)= Z^{2^4 4^4}(\tau,z)\ ,\\
&Z^{\rho_2} (\tau,z)= Z^{1^8 2^8}(\tau,z)\ .
\end{align*}
The multiplicative seed, i.e., $Z^{2^4 4^4}(\tau,z)$ is a Jacobi form of $\Gamma_0(8)$. The cusps and other data at level $8$ are the following:
\begin{center}
\begin{tabular}{|c|c|c|c|c|}
\hline
$f/e$& $i\infty$ & $1/2$&$1/4$ & $0/1$\\
\hline
$h_e$ & 1&2&1& 8\\
\hline
$N_e$& 1& 4&2&8\\
\hline
\end{tabular}
\end{center}
Following the computation similar to the case of cycle shape $1^4 2^2 4^4$, the Borcherds lift leads to the following product formula:
\begin{multline}\label{multiplicative4a}
\Phi_2^{2^4 4^4} (\mathbf{Z}) = qrs\prod_{(n,\ell,m)>0} \left(1-q^n r^\ell s^m\right)^{f_1(nm,\ell)} \left(1-(q^n r^\ell s^m)^2\right)^{\frac12(f_2(nm,\ell)-f_1(nm,\ell))}\\
\left(1-(q^n r^\ell s^m)^4\right)^{\frac14(f_0(nm,\ell)-f_2(nm,\ell))} \ .
\end{multline}
To see if it can be written as the Borcherds product as stated in theorem \ref{CGproduct} we need to compute the Fourier-Jacobi coefficients of the multiplicative seed at the cusps about $i\infty$ and $1/4$ as the Jacobi forms have integer coefficients for these cusps.
\begin{align*}
&\pi_{FE} \Big( Z^{ 2^4 4^4}|M_{\frac14} \Big) = -Z^{2^4 4^4}\ .
\end{align*}
Using this and the product formula for Borcherds product given by Cl\'ery- Gritsenko, we get
\begin{multline}
B_{Z^{2^4 4^4}}(\mathbf{Z})= \prod_{(n,\ell,m)>0} \Big(1- q^n r^\ell s^m\Big)^{f_1(nm,\ell)} \Big(1- (q^n r^\ell s^m)^2\Big)^{-\frac12 f_1(nm,\ell)} \ .
\end{multline}
This product formula does not give all the terms that appear in the product form obtained from the Borcherds lift given in Eq. \eqref{multiplicative4a}. The missing terms can be accounted for by additional terms leading to
\begin{align}\label{modular4a}
\Big( \Phi_2^{2^4 4^4}(\mathbf{Z})\Big)^2 = B_{2Z^{2^4 4^4}}(\mathbf{Z}) B_{Z^{1^8 2^8}}(2\mathbf{Z})\ .
\end{align}
Therefore $ \Big( \Phi_2^{2^4 4^4}(\mathbf{Z})\Big)^2$ is a meromorphic Siegel Modular Form of weight $4$ at level 8. The zeroth coefficient of the Siegel Modular Form, i.e., the $m=0$ term is given by the additive seed $\phi_{2,1}(\tau,z) = \frac{\theta_1(\tau,z)^2}{\eta(\tau)^6} \eta(2\tau)^4\eta(4\tau)^4$.

\subsubsection{Cycle shape $\mathbf{3^8}$}

This is an element of order $3$. The twining genera which contribute to the multiplicative lifts are
\begin{align*}
&Z^{\rho_0} (\tau,z)= Z^{1^{24}}(\tau,z),\\
&Z^{\rho} (\tau,z)= Z^{\rho_2}(\tau,z)= Z^{3^8}(\tau,z)\ .
\end{align*}
The corresponding multiplicative seed $Z^{3^8}(\tau,z)$ is a Jacobi form of $\Gamma_0(9)$. The cusps and other data for $\Gamma_0(9)$ are the following:
\begin{center}
\begin{tabular}{|c|c|c|c|c|}
\hline
$f/e$& $i\infty$ & $1/2$&$1/3$ & $0/1$\\
\hline
$h_e$ & 1&1&1&9\\
\hline
$N_e$& 1& 3&3&9\\
\hline
\end{tabular}
\end{center}
Cusps about $2/3$ can be computed by $ M_{\frac23} = ST^{-1}ST^{2}S$. The cusp about $0$ does not have integer power in $q$ and hence does not contribute. The Borcherds lift given leads to the product formula given by
\begin{align}\label{multiplicative3b}
\Phi_2^{3^8} = qrs \prod_{(n,\ell,m)>0} \Big( 1- q^n r^\ell s^m\Big)^{f_1(nm,\ell)} \Big( 1- (q^n r^\ell s^m)^3\Big)^{\frac13(f_0(nm.\ell)-f_1(nm,\ell)} \ .
\end{align}
To see if it can be written as the Borcherds product as stated in theorem \ref{CGproduct} we need to compute the coefficients at the cusps about $i\infty$, $1/3$ and $2/3$. The contributions from the cusps at $1/3$ and $2/3$ have cube roots of unity. However, the contributions of the two cusps add in the product formula to give integral coefficients. One gets
\begin{align*}
&\pi_{FE} \Big( Z^{ 3^8}|M_{\frac13} + Z^{3^8}|M_{\frac23} \Big) = - Z^{3^8}\ .
\end{align*}
Using this and the product formula for Borcherds product given by Cl\'ery- Gritsenko, we get
\begin{multline}
B_{Z^{3^8}}(\mathbf{Z})=\prod_{(n,\ell,m)>0} \Big(1- q^n r^\ell s^m\Big)^{f_1(nm,\ell)} \Big(1- (q^n r^\ell s^m)^3\Big)^{-\frac13 f_1(nm,\ell)} \ .
\end{multline}
This does not give all the terms that appear in the product form obtained from the Borcherds lift given in Eq. \eqref{multiplicative3b}. The missing terms can be accounted for by additional terms leading to
\begin{align}
\Big( \Phi_2^{3^8}(\mathbf{Z})\Big)^3 = B_{3Z^{4^6}}(\mathbf{Z}) B_{Z^{1^{24}}}(\mathbf{3Z})\ .
\end{align}
Therefore $\Big( \Phi_2^{3^8}(\mathbf{Z})\Big)^3$ is a meromorphic Siegel Modular Form of weight $6$ at level $9$. The zeroth coefficient of the Siegel Modular Form, i.e., the $m=0$ term is given by the additive seed $\phi_{2,1}(\tau,z) = \frac{\theta_1(\tau,z)^2}{\eta(\tau)^6} \eta(3\tau)^8$.

\subsubsection{Cycle shape $\mathbf{4^6}$}

This is an element of order $4$. The twining genera which contribute to the Siegel Modular Forms are
\begin{align*}
&Z^{\rho_0} (\tau,z)= Z^{1^{24}}(\tau,z),\\
&Z^{\rho} (\tau,z)= Z^{\rho_3}(\tau,z)= Z^{4^6}(\tau,z),\\
&Z^{\rho_2} (\tau,z)= Z^{ 2^{12}}(\tau,z)\ .
\end{align*}
The multiplicative seed $Z^{4^6}(\tau,z)$ is a Jacobi form of $\Gamma_0(16)$. The cusps and other data at level $16$ are the following:
\begin{center}
\begin{tabular}{|c|c|c|c|c|c|c|}
\hline
$f/e$& $i\infty$ & $1/2$&$1/4$ &3/4& 1/8 & $0/1$\\
\hline
$h_e$ & 1&4&1&1&1&16\\
\hline
$N_e$& 1& 8&4&4&2&16\\
\hline
\end{tabular}
\end{center}
Cusp about $3/4$ can be computed via following transformation:
$$M_{\frac34} = ST^{-1}ST^{2}S T^{-1}S\ .$$
Following the computation similar to the case of cycle shape $1^4 2^2 4^4$, the Borcherds lift leads to the following product formula:
\begin{multline}\label{multiplicative4c}
\Phi_1^{4^6} (\mathbf{Z}) = qrs\prod_{(n,\ell,m)>0} \left(1-q^n r^\ell s^m\right)^{f_1(nm,\ell)} \left(1-(q^n r^\ell s^m)^2\right)^{\frac12(f_2(nm,\ell)-f_1(nm,\ell))}\\
\left(1-(q^n r^\ell s^m)^4\right)^{\frac14(f_0(nm,\ell)-f_2(nm,\ell))} \ .
\end{multline}
 To see if it can be written as the Borcherds product as stated in theorem \ref{CGproduct}, we need to compute the coefficients at cusps about $i\infty$, ($1/4, 3/4$) and $1/8$. We need to pair up cusps with identical $(h_e, N_e)$. Computing these cusps, we get
$$ \Big( Z^{ 4^6}|M_{\frac14} \Big) = -i Z^{ 4^6}, \text{ and } \Big( Z^{ 4^6}|M_{\frac34} \Big) = i Z^{ 4^6}\ .$$
Then one obtains
\begin{align*}
&\pi_{FE} \Big( Z^{ 4^6}|M_{\frac14} + Z^{ 4^6}|M_{\frac34} \Big) = 0,\\
&\pi_{FE} \Big(Z^{ 4^6}|M_{\frac18} \Big) = -Z^{ 4^6}\ .
\end{align*}
Using these and the product formula for Borcherds product given by Cl\'ery- Gritsenko, we get
\begin{multline}
B_{Z^{4^6}}(\mathbf{Z})= \prod_{(n,\ell,m)>0} \Big(1- q^n r^\ell s^m\Big)^{f_1(nm,\ell)} \Big(1- (q^n r^\ell s^m)^2\Big)^{-\frac12 f_1(nm,\ell)} \ .
\end{multline}
This does not give all the terms that appear in the product form obtained from the Borcherds lift given in Eq. \eqref{multiplicative4c}. The missing terms can be accounted for by additional terms leading to
\begin{align}
\Big( \Phi_1^{4^6}(\mathbf{Z})\Big)^4 = B_{4Z^{4^6}}(\mathbf{Z}) B_{2Z^{ 2^{12}}}(2\mathbf{Z})B_{Z^{1^{24}}}(\mathbf{4Z})\ .
\end{align}
Therefore $\Big( \Phi_1^{4^6}(\mathbf{Z})\Big)^4$ is a modular form of weight $4$ at level $16$. The zeroth coefficient of the Siegel Modular Form, i.e., the $m=0$ term is given by the additive seed $\phi_{1,1}(\tau,z)= \frac{\theta_1(\tau,z)^2}{\eta(\tau)^6} \eta(4\tau)^6$.

\subsubsection{Cycle shape $\mathbf{6^4}$}

This is an element of order $6$ corresponding to the conjugacy class $6b$. The twining genera which contribute to the multiplicative lift are
\begin{align*}
&Z^{\rho_0} (\tau,z)= Z^{1^{24}}(\tau,z),\\
&Z^{\rho} (\tau,z)=Z^{\rho_5} (\tau,z)= Z^{6^4}(\tau,z),\\
&Z^{\rho_2} (\tau,z)=Z^{\rho_4} (\tau,z)= Z^{3^8}(\tau,z),\\
&Z^{\rho_3} (\tau,z)= Z^{2^{12}}(\tau,z)\ .
\end{align*}
The corresponding multiplicative seed, $ Z^{6^4}(\tau,z)$ is a Jacobi form of $\Gamma_0(36)$. The cusps and other data at level $36$ are the following:
\begin{center}
\begin{tabular}{|c|c|c|c|c|c|c|c|c|c|c|c|c|}
\hline
c& $i\infty$&0 & $1/2$ & $1/3$ & $2/3$ & $1/4$ & $1/6$ & $5/6$ & $1/9$ & $1/12$ & $5/12$ & $1/18$\\
\hline
$h_e$&1 & 36&9&4&4&9&1&1&4&1&1&1\\
\hline
$N_e$ & 1& 36&18&12&12&9&6&6&4&3&3&2\\
\hline
\end{tabular}
\end{center}
The cusps about $5/6$ and $5/12$ are respectively given by $M_{\frac56}= -ST^{-1}ST^5S$ and $M_{\frac5{12}} = ST^{-2}ST^2ST^{-2}S$. Following the computation similar to the cycle shape $1^22^23^26^2$, the multiplicative lift leads to the product formula given by
\begin{multline}\label{multiplicative6b}
\Phi_0^{6^4} (\mathbf{Z}) = qrs\prod_{(n,\ell,m)>0} \left(1-q^n r^\ell s^m\right)^{f_1 (nm,\ell)} \left(1-(q^n r^\ell s^m)^2\right)^{\frac12(f_2(nm,\ell)-f_1(nm,\ell))}\\ \left(1-(q^n r^\ell s^m)^3\right)^{\frac13(f_3(nm,\ell)-f_1(nm,\ell))} \left(1-(q^n r^\ell s^m)^6\right)^{\frac16(f_0(nm,\ell)+f_1(nm,\ell)-f_2(nm,\ell)-f_3(nm,\ell))} \ .
\end{multline}
Now to see if can be written as the Borcherds product as stated in theorem \ref{CGproduct} we need to compute the coefficients at cusps about $i\infty$, $(1/6,5/6)$, $(1/12,5/12)$ and $1/18$. We need to pair up cusps with identical values of $h_e$ and $N_e$ to get integral coefficients. Then from straightforward calculations one obtains
\begin{align*}
\pi_{FE} \Big( Z^{ 6^4}|M_{\frac16} + Z^{ 6^4}|M_{\frac56} \Big) &= Z^{ 6^4},\\
\pi_{FE} \Big( Z^{ 6^4}|M_{\frac1{12}} + Z^{ 6^4}|M_{\frac5{12}} \Big) &= - Z^{ 6^4}, \\
\pi_{FE} \Big( Z^{ 6^4}|M_{\frac1{18}}\Big)&=- Z^{ 6^4}\ .
\end{align*}
Using these and the product formula for Borcherds product given by Cl\'ery- Gritsenko, we get
\begin{multline}
B_{Z^{6^4}}(\mathbf{Z})= \prod_{(n,\ell,m)>0} \Big(1- q^n r^\ell s^m\Big)^{f_1(nm,\ell)} \Big(1- (q^n r^\ell s^m)^2\Big)^{-\frac12 f_1(nm,\ell)} \\
\Big(1- (q^n r^\ell s^m)^3\Big)^{-\frac13 f_1(nm,\ell)} \Big(1- (q^n r^\ell s^m)^6\Big)^{\frac16 f_1(nm,\ell)} \ .
\end{multline}
This does not give all the terms that appear in the product form given in Eq. \eqref{multiplicative6b}. The missing terms can be accounted for by additional terms leading to
\begin{align}
\Big( \Phi_0^{6^4}(\mathbf{Z})\Big)^6 = B_{6Z^{6^4}}(\mathbf{Z}) B_{3Z^{ 3^8}}(2\mathbf{Z}) B_{2Z^{ 2^{12}}}(3\mathbf{Z}) B_{Z^{1^{24}}}(6\mathbf{Z})\ .
\end{align}
$\Big( \Phi_0^{6^4}(\mathbf{Z})\Big)^6$ is a meromorphic Siegel Modular Form of weight zero at level 36. The zeroth coefficient of the Siegel Modular Form, i.e., the $m=0$ term is given by the additive seed $\phi_{0,1}(\tau,z) =\frac{\theta_1(\tau,z)^2}{\eta(\tau)^6} \eta(6\tau)^4$.

\subsubsection{Cycle shape $\mathbf{2^2\cdot 10^2}$}

This is an element of order $10$. The twining genera which contribute to the multiplicative lift are
\begin{align*}
&Z^{\rho_0}(\tau,z)= Z^{1^{24}}(\tau,z),\\
&Z^{\rho} (\tau,z)=Z^{\rho_3} (\tau,z)= Z^{\rho_7} (\tau,z)= Z^{\rho_9} (\tau,z)= Z^{2^2 10^2}(\tau,z),\\
&Z^{\rho_2} (\tau,z)=Z^{\rho_4}(\tau,z)=Z^{\rho_6}(\tau,z)=Z^{\rho_8}(\tau,z)=Z^{1^4 5^4}(\tau,z),\\
&Z^{\rho_5} (\tau,z)= Z^{2^{12}}(\tau,z)\ .
\end{align*}
The multiplicative seed, $Z^{2^2 10^2}(\tau,z)$ is a Jacobi form of $\Gamma_0(20)$. The cusps at level $20$ are the following:
\begin{center}
\begin{tabular}{|c|c|c|c|c|c|c|}
\hline
c& $i\infty$ & $1/2$ & $1/4$ & $1/5$ & $1/10$&0/1 \\
\hline
$h_e$&1 & 5&5&4&1&20\\
\hline
$N_e$ & 1& 10&5&4&2&20\\
\hline
\end{tabular}
\end{center}
The Borcherds lift given by Eq. \eqref{productformulamain} leads to the following product formula:
\begin{multline}
\Phi_0^{2^2 10^2} (\mathbf{Z}) = qrs\prod_{(n,\ell,m)>0} \left(1-q^n r^\ell s^m\right)^{c^0 (nm,\ell) + c^1(nm,\ell) - c^2(nm,\ell) - c^5(nm,\ell)} \\
\left(1-(q^n r^\ell s^m)^2\right)^{c^5(nm,\ell) - c^1(nm,\ell)} \left(1-(q^n r^\ell s^m)^5\right)^{c^2(nm,\ell) - c^1(nm,\ell)} \left(1-(q^n r^\ell s^m)^{10}\right)^{c^1(nm,\ell)}\ .
\end{multline}
Then using Eq. \eqref{inverseFJ}, i.e.,
\begin{align}
&c^0 = \frac1{10}(f_0+ 4 f_1+4f_2 +f_5), \hspace{0.2cm} c^1 = \frac1{10}(f_0+f_1 - f_2 -f_5)\ ,\\
&c^2 = \frac1{10}(f_0- f_1- f_2 +f_5), \hspace{0.2cm} c^5 = \frac1{10}(f_0- 4f_1 +4 f_2 -f_5)\ ,
\end{align}
we can rewrite the above product formula as
\begin{multline}\label{multiplicative10a}
\Phi_0^{2^210^2} (\mathbf{Z}) = qrs\prod_{(n,\ell,m)>0} \Big(1-q^n r^\ell s^m \Big)^{f_1(nm,\ell)} \Big(1-(q^n r^\ell s^m)^2 \Big)^{\frac12(f_2(nm,\ell)-f_1(nm,\ell))} \\ \Big(1-(q^n r^\ell s^m)^5 \Big)^{\frac15(f_5(nm,\ell)-f_1(nm,\ell))} \Big(1-(q^n r^\ell s^m)^{10} \Big)^{\frac1{10}(f_0(nm,\ell)+f_1(nm,\ell)- f_2(nm,\ell) - f_5(nm,\ell))}\ .
\end{multline}
Now to see if can be written as the Borcherds product as stated in theorem \ref{CGproduct} we need to compute the coefficients at the cusps about $i\infty$, $1/2$, $1/4$ and $1/10$ as other cusps do not have integral power in $q$.
Then from straightforward calculations one obtains
\begin{align*}
\pi_{FE} \Big( Z^{ 2^210^2}|M_{\frac12} \Big) &=\frac15\left( Z^{ 2^210^2} - Z^{ 2^{12}}\right),\\
\pi_{FE} \Big( Z^{ 2^210^2}|M_{\frac14} \Big) &=-\frac15\left( Z^{ 2^210^2} - Z^{ 2^{12}}\right),\\
\pi_{FE} \Big( Z^{2^210^2}|M_{\frac1{10}}\Big)&=- Z^{ 2^2 10^2}\ .
\end{align*}
Using these and the product formula for Borcherds product given by Cl\'ery- Gritsenko, we get
\begin{multline}
B_{Z^{ 2^2 10^2}}(\mathbf{Z})= \prod_{(n,\ell,m)>0} \Big(1- q^n r^\ell s^m\Big)^{f_1(nm,\ell)} \Big(1- (q^n r^\ell s^m)^2\Big)^{-\frac12 f_1(nm,\ell)} \\
\Big(1- (q^n r^\ell s^m)^5\Big)^{\frac15(f_5(nm,\ell)- f_1(nm,\ell)} \Big(1- (q^n r^\ell s^m)^{10}\Big)^{\frac1{10}(f_1(nm,\ell)-f_5(nm,\ell))} \ .
\end{multline}
This does not give all the terms that appear in the product form given in Eq. \eqref{multiplicative10a}. The missing terms can be accounted for by additional terms leading to
\begin{align}\label{modular10a}
\Big( \Phi_0^{ 2^2 10^2}(\mathbf{Z})\Big)^2 = B_{2Z^{ 2^2 10^2}}(\mathbf{Z}) B_{Z^{ 1^4 5^4}}(2\mathbf{Z}) \ .
\end{align}
$\Big( \Phi_0^{2^2 10^2}(\mathbf{Z})\Big)^2$ is a meromorphic Siegel Modular Form of weight zero at level 20. The zeroth coefficient of the Siegel Modular Form, i.e., the $m=1$ term is given by the additive seed $\phi_{0,1}(\tau,z) = \frac{\theta_1(\tau,z)^2}{\eta(\tau)^6} \eta(2\tau)^2 \eta(10\tau)^2$.

\subsubsection{Cycle shape $\mathbf{2\cdot 4\cdot 6 \cdot 12}$}

This is an element of order $12$. The twining genera which contribute to the multiplicative lift are
\begin{align*}
&Z^{\rho_0} (\tau,z)=Z^{1^{24}}(\tau,z),\\
&Z^{\rho} (\tau,z)=Z^{\rho_5} (\tau,z)= Z^{\rho_7} (\tau,z)= Z^{\rho_{11}} (\tau,z)= Z^{2^1 4^1 6^1 12^1} (\tau,z),\\
&Z^{\rho_2} (\tau,z)= Z^{\rho_{10}} (\tau,z)= Z^{1^2 2^2 3^2 6^2 }(\tau,z),\\
&Z^{\rho_3} (\tau,z)=Z^{\rho_9} (\tau,z)= Z^{2^4 4^4}(\tau,z),\\
&Z^{\rho_4} (\tau,z) =Z^{\rho_8} (\tau,z)= Z^{1^6 3^6},\\
&Z^{\rho_6} (\tau,z) = Z^{1^8 2^8}\ .
\end{align*}
The corresponding multiplicative seed, $Z^{2^1 4^1 6^1 12^1} (\tau,z)$ is a Jacobi form of $\Gamma_0(24)$. The cusps at level $24$ are the following:
\begin{center}
\begin{tabular}{|c|c|c|c|c|c|c|c|c|c|c|c|c|}
\hline
c& $i\infty$&0 & $1/2$ & $1/3$ & $1/4$ & $1/6$ & $1/8$ & $1/12$ \\
\hline
$h_e$&1 & 24& 6 &8&3&2&3&1\\
\hline
$N_e$ & 1& 24&12&8&6&4&3&2\\
\hline
\end{tabular}
\end{center}
The Borcherds lift given by Eq. \eqref{productformulamain} leads to the following product formula:
\begin{multline}
\Phi_0^{2^1 4^1 6^1 12^1} (\mathbf{Z}) = qrs\prod_{(n,\ell,m)>0} \left(1-q^n r^\ell s^m\right)^{c^0 (nm,\ell) + c^2(nm,\ell) - c^4(nm,\ell) - c^6(nm,\ell)} \\
\left(1-(q^n r^\ell s^m)^2\right)^{c^6(nm,\ell) - c^2(nm,\ell) - c^3(nm,\ell) + c^1(nm,\ell)}\left(1-(q^n r^\ell s^m)^3\right)^{c^4(nm,\ell) - c^2(nm,\ell)} \\
\left(1-(q^n r^\ell s^m)^4\right)^{c^3(nm,\ell)-c^1(nm,\ell)} \left(1-(q^n r^\ell s^m)^6\right)^{c^2(nm,\ell)-c^1(nm,\ell)}\left(1-(q^n r^\ell s^m)^{12}\right)^{c^1(nm,\ell)}\ .
\end{multline}
Then using Eq. \eqref{inverseFJ}, i.e.,
\begin{align*}
&c^0 = \frac1{12}(f_0+ 4 f_1 +2 f_2 +2 f_3 +2f_4 + f_6), \hspace{0.2cm} c^1 = \frac1{12}(f_0+f_2 -f_4 -f_6)\ ,\\
&c^2 = \frac1{12}(f_0+ 2 f_1 - f_2 -2 f_3 -f_4 + f_6), \hspace{0.2cm} c^3 = \frac1{12}(f_0- 2f_2 +2f_4 -f_6)\ ,\\
&c^6= \frac1{12}(f_0- 4 f_1 +2 f_2 -2 f_3 +2f_4 + f_6), \hspace{0.2cm} c^4 = \frac1{12}(f_0-2f_1-f_2+2f_3 -f_4 + f_6)\ ,
\end{align*}
we can rewrite the above product formula as
\begin{multline}\label{multiplicative12a}
\Phi^{2^1 4^1 6^1 12^1}_0 (\mathbf{Z}) = qrs\prod_{(n,\ell,m)>0} \Big(1-q^n r^\ell s^m \Big)^{f_1(nm,\ell)} \Big(1-(q^n r^\ell s^m)^2 \Big)^{\frac12(f_2(nm,\ell)-f_1(nm,\ell))} \\ \Big(1-(q^n r^\ell s^m)^3 \Big)^{\frac13(f_3(nm,\ell)-f_1(nm,\ell))} \Big(1-(q^n r^\ell s^m)^4\Big)^{\frac14 (- f_2(nm,\ell) + f_4(nm,\ell))}\\
\Big(1-(q^n r^\ell s^m)^6 \Big)^{\frac13(f_1(nm,\ell)-f_2(nm,\ell)-f_3(nm,\ell)+f_6(nm,\ell))}\\ \Big(1-(q^n r^\ell s^m)^{12}\Big)^{\frac1{12} (f_0(nm,\ell) +f_2(nm,\ell) - f_4(nm,\ell) - f_6(nm,\ell))}\ .
\end{multline}
Now to see if can be written as the Borcherds product as stated in theorem \ref{CGproduct} we need to compute the coefficients at different cusps.
Only the cusps about $1/4, 1/8$ and $1/12$ has integer power of $q$ in the Fourier expansion. The other cusps do not contribute. Then from straightforward calculations one obtains
\begin{align*}
\pi_{FE} \Big( Z^{2^1 4^1 6^1 12^1}|M_{\frac14} \Big) &= \frac13\left( Z^{2^1 4^1 6^1 12^1} - Z^{2^4 4^4}\right),\\
\pi_{FE} \Big( Z^{ 2^1 4^1 6^1 12^1}|M_{\frac18} \Big) &= - \frac13\left( Z^{2^1 4^1 6^1 12^1} - Z^{2^4 4^4}\right), \\
\pi_{FE} \Big( Z^{ 2^1 4^1 6^1 12^1}|M_{\frac1{12}}\Big)&=- Z^{ 2^1 4^1 6^1 12^1}\ .
\end{align*}
Using these and the product formula for Borcherds product given by Cl\'ery- Gritsenko, we get
\begin{multline}
B_{Z^{ 2^1 4^1 6^1 12^1}}(\mathbf{Z})=\prod_{(n,\ell,m)>0} \Big(1- q^n r^\ell s^m\Big)^{f_1(nm,\ell)} \Big(1- (q^n r^\ell s^m)^2\Big)^{-\frac12 f_1(nm,\ell)} \\
\Big(1- (q^n r^\ell s^m)^3\Big)^{-\frac13( f_1(nm,\ell) - f_3(nm,\ell) )} \Big(1- (q^n r^\ell s^m)^6\Big)^{\frac16( f_1(nm,\ell)- f_3(nm,\ell))} \ .
\end{multline}
This does not give all the terms that appear in the product form given in Eq. \eqref{multiplicative12a}. The missing terms can be accounted for by additional terms leading to
\begin{align}
\Big( \Phi_0^{ 2^1 4^1 6^1 12^1}(\mathbf{Z})\Big)^2 = B_{2Z^{ 2^1 4^1 6^1 12^1}}(\mathbf{Z}) B_{Z^{ 1^2 2^2 3^2 6^2}}(2\mathbf{Z}) \ .
\end{align}
$\Big( \Phi_0^{ 2^1 4^1 6^1 12^1}(\mathbf{Z})\Big)^2 $ is a meromorphic Siegel Modular Form of weight zero at level 24. The zeroth coefficient of the Siegel Modular Form, i.e., the $m=0$ term is given by the additive seed $\phi_{0,1}(\tau,z) = \frac{\theta_1(\tau,z)^2}{\eta(\tau)^6} \eta(2\tau) \eta(4\tau) \eta(6\tau) \eta(12\tau)$.

\subsubsection{Cycle shape $\mathbf{12^2}$}\label{Details12squared}

This is an element of order $12$. The twining genera which contribute to the multiplicative lift are
\begin{align*}
&Z^{\rho_0} (\tau,z)=Z^{1^{24}}(\tau,z),\\
&Z^{\rho} (\tau,z)=Z^{\rho_5} (\tau,z)=Z^{\rho_7} (\tau,z)=Z^{\rho_{11}}(\tau,z)= Z^{12^2} (\tau,z),\\
&Z^{\rho_2}(\tau,z)=Z^{\rho_{10}} (\tau,z)= Z^{6^4 }(\tau,z),\\
&Z^{\rho_3} (\tau,z)=Z^{\rho_9} (\tau,z)= Z^{ 4^6}(\tau,z),\\
&Z^{\rho_4}(\tau,z) =Z^{\rho_8}(\tau,z)= Z^{ 3^8},\\
&Z^{\rho_6}(\tau,z) = Z^{2^{12}}\ .
\end{align*}
The corresponding multiplicative seed, $Z^{12^2} (\tau,z)$ is a Jacobi form of $\Gamma_0(144)$. The cusps and other data at level $144$ are the following:
\begin{center}
\begin{tabular}{|c|c|c|c|c|c|c|c|c|c|c|c|c|c|c|c|}
\hline
c& $i\infty$&0 & $1/2$ & $1/3$&2/3 & $1/4$&3/4 & $1/6$&5/6 & $1/8$&1/9 & $1/12$&5/12 \\
\hline
$h_e$&1 & 144& 36 &16&16&9&9&4&4&9&16&1&1\\
\hline
$N_e$ & 1& 144&72&48&48&36&36&24&24&18&16&12&12\\
\hline
\end{tabular}
\end{center}
\begin{center}
\begin{tabular}{|c|c|c|c|c|c|c|c|c|c|c|c|c|c|c|c|}
\hline
c& $7/12$ & $11/12$&1/16 & $1/18$&1/24& $5/24$&1/36 & $7/36$&$1/48$ & $5/48$&1/72 \\
\hline
$h_e$&1 & 1& 9&4&1&1&1&1&1&1&1\\
\hline
$N_e$ & 12& 12&9&8&6&6&4&4&3&3&2\\
\hline
\end{tabular}
\end{center}

\subsubsection{Cycle shape $\mathbf{3\cdot 21}$}

This is an element of order $21$. The twining genera that contribute to the multiplicative lift are
\begin{align*}
&Z^{\rho_0} (\tau,z)= Z^{1^{24}}(\tau,z),\\
&Z^{\rho} (\tau,z)=Z^{\rho_2} (\tau,z) = Z^{\rho_4} (\tau,z)= Z^{\rho_5} (\tau,z)= Z^{\rho_8} (\tau,z)= Z^{\rho_{10}} (\tau,z)= Z^{\rho_{11}} (\tau,z)\\
&\hspace{1cm}= Z^{\rho_{13}} (\tau,z)= Z^{\rho_{16}} (\tau,z)= Z^{\rho_{17}} (\tau,z)= Z^{\rho_{19}} (\tau,z)= Z^{\rho_{20}} (\tau,z)= Z^{3^1 21^1}(\tau,z),\\
& Z^{\rho_3} (\tau,z) = Z^{\rho_6} (\tau,z) = Z^{\rho_9} (\tau,z) = Z^{\rho_{12}} (\tau,z) = Z^{\rho_{15}} (\tau,z) = Z^{\rho_{18}} (\tau,z) = Z^{1^3 7^3}(\tau,z),\\
& Z^{\rho_7} (\tau,z) = Z^{\rho_{14}} (\tau,z) )= Z^{3^8}(\tau,z)\ .
\end{align*}
In this case the multiplicative seed, $Z^{3^1 21^1}(\tau,z)$ is a Jacobi form of $\Gamma_0(63)$. The cusps and other data for $\Gamma_0(63)$ are the following:
\begin{center}
\begin{tabular}{|c|c|c|c|c|c|c|c|c|}
\hline
c& $i\infty$ & $1/3$ & $2/3$ & $1/7$ & $1/9$&1/21&2/21&0/1 \\
\hline
$h_e$&1 & 7&7&9&7&1&1&63\\
\hline
$N_e$ & 1& 21&21&9&7&3&3&63\\
\hline
\end{tabular}
\end{center}
The Borcherds lift given by Eq. \eqref{productformulamain} leads to the following product formula:
\begin{multline}
\Phi_{-1}^{3^1 21^1} (\mathbf{Z}) = qrs\prod_{(n,\ell,m)>0} \left(1-q^n r^\ell s^m\right)^{c^0 (nm,\ell) + c^1(nm,\ell) - c^3(nm,\ell) - c^7(nm,\ell)} \\
\left(1-(q^n r^\ell s^m)^3\right)^{c^7(nm,\ell) - c^1(nm,\ell)}
\left(1-(q^n r^\ell s^m)^7\right)^{c^3(nm,\ell) - c^1(nm,\ell)} \left(1-(q^n r^\ell s^m)^{21}\right)^{c^1(nm,\ell)}\ .
\end{multline}
Then using Eq. \eqref{inverseFJ}, i.e.,
\begin{align*}
&c^0 = \frac1{21}(f_0+ 12 f_1+ 6f_3 + 2f_7), \hspace{0.2cm} c^1 = \frac1{21}(f_0+f_1 - f_3 -f_7)\ ,\\
&c^3 = \frac1{21}(f_0-2 f_1- f_3+2f_7), \hspace{0.2cm} c^7 = \frac1{21}(f_0- 6f_1 +6 f_3 -f_7)\ ,
\end{align*}
we can rewrite the above product formula as
\begin{multline}\label{multiplicative21a}
\Phi_{-1}^{3^1 21^1} (\mathbf{Z}) = qrs\prod_{(n,\ell,m)>0} \Big(1-q^n r^\ell s^m \Big)^{f_1(nm,\ell)} \Big(1-(q^n r^\ell s^m)^3 \Big)^{\frac13(f_3(nm,\ell)-f_1(nm,\ell))} \\ \Big(1-(q^n r^\ell s^m)^7 \Big)^{\frac17(f_7(nm,\ell)-f_1(nm,\ell))} \Big(1-(q^n r^\ell s^m)^{21} \Big)^{\frac1{21}(f_0(nm,\ell)+f_1(nm,\ell)- f_3(nm,\ell) - f_7(nm,\ell))}\ .
\end{multline}
Now to see if can be written as the Borcherds product as stated in theorem \ref{CGproduct} we need to compute the coefficients at different cusps.
We need to consider the contributions from the cusps at $i\infty$, $(1/3,2/3)$, $1/9$, and $(1/21,2/21)$ as other cusps do not have integral power in $q$. Here again, we need to pair up cusps with identical $h_e$ and $N_e$ to obtain integral coefficients.
Then from straightforward calculations one obtains
\begin{align*}
\pi_{FE} \Big( Z^{3^121^1}|M_{\frac13} + Z^{3^121^1}|M_{\frac23} \Big) &=\frac17 (Z^{ 3^1 21^1} - Z^{ 3^8}), \\
\pi_{FE} \Big( Z^{3^121^1}|M_{\frac19} \Big) &=-\frac17 (Z^{ 3^1 21^1} - Z^{ 3^8}), \\
\pi_{FE} \Big( Z^{3^121^1}|M_{\frac1{21}} + Z^{3^121^1}|M_{\frac2{21}} \Big) &=-Z^{ 3^1 21^1}\ .
\end{align*}
Using these and the product formula for Borcherds product given by Cl\'ery- Gritsenko, we get
\begin{multline}
B_{Z^{ 3^1 21^1}}(\mathbf{Z})= \prod_{(n,\ell,m)>0} \Big(1- q^n r^\ell s^m\Big)^{f_1(nm,\ell)} \Big(1- (q^n r^\ell s^m)^2\Big)^{-\frac13 f_1(nm,\ell)} \\
\Big(1- (q^n r^\ell s^m)^7\Big)^{\frac17(f_7(nm,\ell)- f_1(nm,\ell))} \Big(1- (q^n r^\ell s^m)^{21}\Big)^{\frac1{21}(f_1(nm,\ell)-f_7(nm,\ell))} \ .
\end{multline}
This does not give all the terms that appear in the product form given in Eq. \eqref{multiplicative21a}. The missing terms can be accounted for by additional terms leading to
\begin{align}\label{modular21ab}
\Big( \Phi_{-1}^{3^1 21^1}(\mathbf{Z})\Big)^3= B_{3 Z^{ 3^1 21^1}}(\mathbf{Z}) B_{Z^{ 1^3 7^3}}(3\mathbf{Z})\ .
\end{align}
$\Big( \Phi_{-1}^{3^1 21^1}(\mathbf{Z})\Big)^3$ is a meromorphic Siegel Modular Form of weight $(-3)$ at level 63. The zeroth coefficient of the Siegel Modular Form, i.e., the $m=0$ term is given by the additive seed $\phi_{-1,1}(\tau,z) =\frac{\theta_1(\tau,z)^2}{\eta(\tau)^6} \eta(3\tau)\eta(21\tau)$.

Unlike the type-I case, none of the conjugacy class has a cycle shape with cycle length one. In all these cases, the multiplicative lift can be expressed as a product of multiple rescaled Borcherds products. This proves that the multiplicative lift can be used to construct a genus-two Siegel Modular Form for all the conjugacy classes of $M_{24}$, including those for which the additive construction is not known.

\bibliographystyle{utphys}
\bibliography{refs}
\end{document}